\documentclass[11pt]{article}
\usepackage[T1]{fontenc}
\usepackage[margin=1in]{geometry}
\usepackage{amsfonts,amsmath,amsthm,amssymb,mathtools}  %

\usepackage{nicematrix}

\mathtoolsset{centercolon}
\usepackage{xfrac,nicefrac}
\usepackage{mathdots}
\usepackage{mleftright}  %

\usepackage{xspace}
\xspaceaddexceptions{]\}}  %
\usepackage{regexpatch}

\usepackage{bm,bbm,dsfont}  %
\usepackage{caption}
\usepackage[normalem]{ulem}
\usepackage{enumitem}

\usepackage{graphicx}
\usepackage{float}
\usepackage{subcaption}  %
\usepackage{tcolorbox}
\usepackage{tikz}
\usetikzlibrary{decorations.pathreplacing}
\usetikzlibrary{calc}
\usetikzlibrary{positioning}
\usetikzlibrary{arrows.meta}
\usetikzlibrary{math}
\usetikzlibrary{patterns}

\usepackage[linesnumbered,boxed,ruled,vlined]{algorithm2e}

\SetCommentSty{mycommfont}

\usepackage{thmtools,thm-restate}
\usepackage[colorlinks,citecolor=blue,linkcolor=blue,urlcolor=red]{hyperref}

\theoremstyle{plain}

\newtheorem{theorem}{Theorem}[section]  %
\newtheorem{lemma}[theorem]{Lemma}

\newtheorem{corollary}[theorem]{Corollary}

\newtheorem{claim}[theorem]{Claim}
\newtheorem{conjecture}[theorem]{Conjecture}
\newtheorem{hypothesis}[theorem]{Hypothesis}
\theoremstyle{definition}  %

\newenvironment{proofof}[1]{\begin{proof}[Proof of #1]}{\end{proof}}

\usepackage[capitalise]{cleveref}
\crefname{algocf}{Algorithm}{Algorithms}
\Crefname{algocf}{Algorithm}{Algorithms}
\crefname{claim}{Claim}{Claims}
\Crefname{claim}{Claim}{Claims}

\newfloat{Distribution}{htbp}{loa}
\crefname{Distribution}{Distribution}{Distributions}
\Crefname{Distribution}{Distribution}{Distributions}
\newfloat{Protocol}{htbp}{loa}
\crefname{Protocol}{Protocol}{Protocols}
\Crefname{Protocol}{Protocol}{Protocols}

\SetKwProg{Function}{Function}{:}{}

\SetKwProg{CodeBlock}{}{}{}
\SetKwProg{Repeat}{Repeat}{:}{}

\DeclarePairedDelimiter{\bk}{(}{)}
\DeclarePairedDelimiter{\Bk}{[}{]}
\DeclarePairedDelimiter{\BK}{\{}{\}}

\DeclarePairedDelimiterX\mysetbase[2]{\lbrace}{\rbrace}{#1\,\delimsize\vert\,#2}
\NewDocumentCommand{\myset}{sO{}m m}{%
  \IfBooleanTF{#1}%
    {\mysetbase*{#3}{#4}}%
    {\mysetbase[#2]{#3}{#4}}%
}

\DeclareMathOperator*{\E}{\mathbb{E}}

\let\Pr\PrAux
\DeclareMathOperator{\poly}{poly}

\DeclareMathOperator*{\ind}{\mathbbm{1}}

\newcommand{\F}{\mathbb{F}}

\renewcommand{\tilde}{\widetilde}

\newcommand{\defeq}{\coloneqq}
\newcommand{\eps}{\varepsilon}

\renewcommand{\epsilon}{\eps}
\newcommand{\Patrascu}{\textup{P{\v{a}}tra{\c{s}}cu}\xspace}

\newcommand{\defn}[1]{\emph{\boldmath\textbf{#1}}}

\newcommand{\numberthis}{\addtocounter{equation}{1}\tag{\theequation}}

\newcommand{\matwrap}[1]{{\begin{matrix}#1\end{matrix}}}

\usepackage{regexpatch}
\makeatletter
\xpatchcmd\thmt@restatable{%
\csname #2\@xa\endcsname\ifx\@nx#1\@nx\else[{#1}]\fi
}{%
\ifthmt@thisistheone
\csname #2\@xa\endcsname\ifx\@nx#1\@nx\else[{#1}]\fi
\else
\csname #2\@xa\endcsname[{Restated}]
\fi}{}{}
\makeatother

\newcommand{\OPT}{\mathbf{OPT}}
\newcommand{\dmin}{M_{\min}}
\newcommand{\dmax}{M_{\max}}

\usepackage{physics}
\renewcommand{\vec}[1]{\vb{#1}}

\newcommand{\GF}{\textup{GF}}
\newcommand{\mfix}{m_{\textup{fix}}}
\newcommand{\mconv}[1][s]{m_{\text{conv}}^{(#1)}}
\newcommand{\mretr}[1][s]{m_{\text{retr}}^{(#1)}}

\usepackage{xargs}

\title{Optimal Static Dictionary with Worst-Case Constant Query Time}
\author{Yang Hu\thanks{Institute for Interdisciplinary Information Sciences, Tsinghua University. \texttt{y-hu22@mails.tsinghua.edu.cn}.}
\and
Jingxun Liang\thanks{Carnegie Mellon University. \texttt{jingxunl@andrew.cmu.edu}.}
\and
Huacheng Yu\thanks{Princeton University. \texttt{yuhch123@gmail.com}. Supported by Simons Junior Faculty Award - AWD1007164.}
\and
Junkai Zhang\thanks{Institute for Interdisciplinary Information Sciences, Tsinghua University. \texttt{zhangjk22@mails.tsinghua.edu.cn}.}
\and
Renfei Zhou\thanks{Carnegie Mellon University. Partially supported by the MongoDB PhD Fellowship. \texttt{renfeiz@andrew.cmu.edu}.}
}
\date{}

\begin{document}

\maketitle

\begin{abstract}
    In this paper, we design a new succinct static dictionary with worst-case constant query time.
    A dictionary data structure stores a set of key-value pairs with distinct keys in $[U]$ and values in $[\sigma]$, such that given a query $x\in [U]$, it quickly returns if $x$ is one of the input keys, and if so, also returns its associated value.
    The textbook solution to dictionaries is hash tables.
    On the other hand, the (information-theoretical) optimal space to encode such a set of key-value pairs is only $\OPT:=\log\binom{U}{n}+n\log \sigma$.

    We construct a dictionary that uses $\OPT+n^{\epsilon}$ bits of space, and answers queries in constant time \emph{in worst case}.
    Previously, constant-time dictionaries are only known with $\OPT+n/\poly\log n$ space~\cite{patrascu2008succincter}, or with $\OPT+n^{\epsilon}$ space but \emph{expected} constant query time~\cite{yu2020nearly}.
    We emphasize that most of the extra $n^{\epsilon}$ bits are used to store 
    \begin{itemize}
        \item a lookup table that does not depend on the input, and
        \item random bits for hash functions.
    \end{itemize}
    The ``main'' data structure only occupies $\OPT+\poly\log n$ bits.
\end{abstract}

\section{Introduction}
Dictionaries are one of the most fundamental data structures, with numerous applications in computer science.
A \defn{static dictionary} stores a set of $n$ pairs $(x_i,v_i)$, where $x_i\in [U]$ are \emph{distinct} keys and $v_i\in[\sigma]$ are the associated values.
Given a query $x$, it returns whether there is a pair $(x_i,v_i)$ with $x_i=x$, and if so, it also returns $v_i$.
The \defn{membership data structures} are the special case of $\sigma = 1$, i.e., a membership data structure stores a set of keys, such that given any $x$, it returns if $x$ is in the set.
In most parts of this paper, we focus on static membership data structures, but we will show that our techniques also generalize to general static dictionaries.\footnote{There is also a dynamic setting for dictionaries, which further requires to support efficient insertions and deletions of a key-value pair $(x, v)$. However, we only focus on the static version in this paper.}

The textbook solution to dictionaries is the hash tables~\cite{clrsbook}.
Pairwise independent hash functions yield very few hash collisions in expectation.
By applying the standard chaining technique to the hash table, the queries can be answered in expected constant time.
Perfect hashing~\cite{fredman1984storing} further eliminates all collisions by utilizing a two-level hashing, and guarantees \emph{worst-case} constant query time.
In general, hash tables take $O(n\log U)$ bits to store the keys.
This is, in fact, suboptimal in space.
Information theoretically, $\log\binom{U}{n}$ bits are sufficient to encode the set, hence, the best possible space for storing the key-value pairs is
\[
    \OPT:=\log\binom{U}{n}+n\log \sigma.
\]
Note that when $U$ and $n$ are close, $n\log U$ can be asymptotically larger than $\log\binom{U}{n}$, and for $U=n^{1+\Theta(1)}$, the two quantities differ by a constant factor.

There has been a rich literature in designing space-efficient membership data structures and dictionaries~\cite{CW79,TY79,Yao81a,fredman1984storing,FNSS92,FN93,FM95,Mil96,MNSW98,BM99,pagh2001low,Pagh01b,BMRV02,patrascu2008succincter,yu2020nearly}.
In particular, it is known how to construct \emph{succinct} dictionaries, i.e., dictionaries using space $\OPT+o(\OPT)$.
The $o(\OPT)$ term in the space bound is usually called the \defn{redundancy}.
The state-of-the-art dictionaries are by P\v{a}tra\c{s}cu~\cite{patrascu2008succincter} achieving redundancy $n/\poly\log n$ and worst-case constant query time, and by Yu~\cite{yu2020nearly} achieving redundancy $n^{\epsilon}$ and expected constant query time.

In this paper, we achieve the best of the previous two works: For worst-case constant query dictionaries, we improve the redundancy from $n / \poly \log n$ to $n^{\epsilon}$.

\begin{restatable}{theorem}{MainTheoremIntro}
  \label{thm:intro}
  In the word RAM model with word size $w = \Theta(\log n)$, there is a static dictionary storing $n$ keys from a universe of size $U \in [2n,\, \poly n]$ and values from a universe of size $\sigma \in [1,\, \poly n]$, using $\OPT + \poly \log n$ bits of space, assuming access to a fixed lookup table of $n^\eps$ bits and a hash function that can be encoded using $n^\eps$ bits, such that:
  \begin{itemize}
  \item The construction algorithm succeeds with high probability in $n$.
  \item The query algorithm runs in worst-case constant time.
  \end{itemize}
\end{restatable}

We remark that the data structure uses randomness to sample a hash function, and this is the only probabilistic part that occurs in the construction algorithm: It must find a ``good'' hash function, and a random hash function is ``good'' with $1-1/\poly n$ probability.
Once a ``good'' hash function is used, the query time is constant in the worst case, independent of the choice of the hash function.
Hence, one could also repeatedly sample hash functions during the construction, until a good hash function is found.
In this way, the construction algorithm can be made Las Vegas.
By storing the lookup table and the hash function as part of the data structure, the total redundancy is still only $n^{\epsilon}$.
Moreover, since the lookup table does not depend on the data and a random hash function works with high probability, if we store at most $\poly n$ dictionaries, the lookup table and the hash function can be shared among all dictionaries, and only one copy needs to be stored.

We also remark that in the \emph{cell-probe} model, where computation is free and we only charge the number of memory accesses, the total redundancy is only $\poly\log n$ by hardwiring the lookup table, and by applying a variant of Newman's theorem~\cite{Newman91}, which allows us to hardwire $\poly n$ hash functions and remember which hash function we are using.

\bigskip

To prove the theorem, we propose a generic way to store variable-length data structures.
Assume we are given a list of data structures $D_i$, where $D_i$ has size $m_i$.
The sizes $m_i$ are not known in advance, and may vary based on the actual input.
We show that these data structures can be jointly stored with \emph{almost no} extra redundancy.
Note that naively concatenating the data structures would require us to know $m_1+\cdots+m_{i-1}$ in order to locate $D_i$, which results in extra redundancy.
To achieve this, we make use of a new \emph{retrieval} data structure.
We give more details in the next section.

\subsection{Related Work}
The FKS perfect hashing scheme~\cite{fredman1984storing} takes $O(n\sqrt{\log n})$ bits to store the hash function.
It was improved to $O(n)$ bits by subsequent works~\cite{SS90,HT01}.
It was also shown~\cite{FNSS92,FN93} that the hash table and the hash function can be jointly stored in space $n\lceil\log U\rceil + n\lceil\log \sigma\rceil$.
Brodnik and Munro~\cite{BM99} constructed the first succinct dictionary, with space $\OPT+O(\OPT/\log\log\log U)$.
It was improved by Pagh~\cite{pagh2001low} to $o(n)$ bits of redundancy, then by P\v{a}tra\c{s}cu~\cite{patrascu2008succincter} to $n/\log^c n$ bits for any constant $c$, and by Yu~\cite{yu2020nearly} to $n^{\epsilon}$ bits for any constant $\epsilon$ but with expected (constant) query time.

In the \emph{bit-probe} model, where the query algorithm can only access one bit of the data structure in constant time,  Buhrman, Miltersen, Radhakrishnan, and Venkatesh~\cite{BMRV02} showed that for $O(\OPT)$ space, the query time must be $\Omega(\log \frac{U}{n})$.
Viola~\cite{Viola12a} showed that for $U=3n$, any bit-probe dictionary with query time $q$ must use space $\OPT+n/2^{O(q)}-\log n$.

For dynamic dictionary, where we further allow insertions and deletions of key-value pairs, the state-of-the-art solution~\cite{bender2022optimal} has $O(n\underbrace{\log\cdots\log}_{k} n)$ bits of redundancy for query time $k$, which has been shown to be optimal~\cite{li2023tight}.

\section{Technical Overview}\label{sec:overview}
In this section, we give an overview of our new dictionary.
We will focus on constructing a membership data structure for simplicity.
To begin with, we further assume that the construction algorithm and the query algorithm have \emph{free} access to random bits.
Recall that the word size $w$ is $\Theta(\log U)=\Theta(\log n)$.
The starting point is a worst-case constant time near-optimal dictionary for $(\poly w)$-sized sets from~\cite{yu2020nearly}.
The data structure has at most ``$1/\poly n$ bits'' of redundancy, which we will clarify later in this section.

We begin with presenting an (over)simplified version of our dictionary that works under several assumptions.
We will remove these assumptions later.
The simplified dictionary first applies a random permutation to all keys, and partitions the universe $[U]$ into $n/B$ buckets of equal size $V:=U/(n/B)$, where $B=\poly\log n$.
Over the random choice of the permutation, each bucket $i$ has $s_i=B\pm B^{2/3}$ keys with high probability.
Assume this happens for every bucket.
In particular, every bucket has $\poly\log n=\poly w$ keys.
Then we can apply~\cite{yu2020nearly} to construct an optimal dictionary $D_i$ for each bucket $i$ with space $\approx \log \binom{V}{s_i}$ bits.
For now, assume $\log \binom{V}{s_i}$ happens to be $(1/\poly n)$-close to the next integer multiple of $w$, so that each $D_i$ occupies an integer number of words and we have wasted a total of $1/\poly n$ bits of space from all $D_i$.
Given a query $x$, the query algorithm applies the random permutation to $x$, gets the bucket that contains $x$, and then queries the corresponding $D_i$.
It might seem that we are done, but the actual challenge is in jointly storing the $n/B$ data structures $D_i$.

The naive approach is to simply concatenate all $D_i$.
But in order to locate $D_i$ within the concatenated data structure, the total size of the first $i-1$ data structures, $|D_1|+\cdots+|D_{i-1}|$, must be computed.
Unfortunately, this partial sum task has a lower bound: It requires $(n/B)/\poly\log n=n/\poly\log n$ bits of redundancy if given any $i$, we want to compute the sum in constant time~\cite{patrascu2010cellprobe, viola2023new}.
The expected-time dictionary of Yu~\cite{yu2020nearly} divides each $D_i$ into two data structures $D_{i,1}, D_{i,2}$ in a careful way, such that $D_{i,1}$ has \emph{fixed} length, and ``most'' queries can be answered by only accessing $D_{i,1}$.
By concatenating all $D_{i,1}$, each of them can be easily located due to their fixed lengths.
Hence, most queries take constant time by only accessing $D_{i,1}$.
All $D_{i,2}$ are then concatenated while equipped with a low redundancy but $(\log n)$-query-time partial sum data structure to locate each component.
The remaining ``hard'' queries, which rely on $D_{i,2}$, will take a longer query time.
Finally, by randomly shifting the queries, each query is ``hard'' only with small probability.
The expected time of any given query is still a constant.
This approach reduces the set of slow queries, but does not help if we want worst-case constant query time.

In this paper, we circumvent this partial-sum lower bound by using a completely different approach to store all $D_i$.
Our simple, yet crucial, observation is that each $D_i$ does not necessarily have to be stored consecutively in memory; moreover, each word of $D_i$ does not even have to be stored explicitly; we only want each word of $D_i$ to be recoverable in constant time, so that the query algorithm on each $D_i$ can be simulated with the same running time.

For now, assume for simplicity that we have a way to store all $s_i$, and each $s_i$ can be recovered in constant time.
Thus, the size of each $D_i$ is determined.
The concrete task is formally the following: Given data structures $D_1,\ldots,D_{n/B}$, we would like to store them with (almost) no redundancy, such that given a pair $(i,j)$, the $j$-th word in $D_i$ can be recovered in constant time.
This is the main question that we solve, and the solution uses a \emph{retrieval data structure}.

\subsection{Retrieval data structures}
A (static) retrieval data structure stores a set of $n$ key-value pairs, where the keys are from a key space $[N]$, and in our application, the value space is $[2^w]$.
The data structure supports \defn{retrieval} queries: Given a key $x\in[N]$, output the value associated with $x$.
The set of keys in the input are called the \defn{valid} keys.
Note that this is different from a dictionary, in that the data structure does not need to check if $x$ is \emph{valid}.
It suffices to return the associated value \emph{assuming} the key is in the input (otherwise, the data structure is allowed to return anything).
The information-theoretical optimal space is $nw$ bits, as we do not encode the set.

By viewing $(i,j)$ as a key and the content of the $j$-th word in $D_i$ as the associated value, the task of recovering the content is exactly what a retrieval data structure does.
However, the state-of-the-art retrieval data structure with almost no redundancy and constant query time only supports 1-bit values~\cite{dietzfelbinger2019constanttime}.
The same idea would cost $\log n$ query time to retrieve a $w$-bit word.
We use the special structure of the set of valid keys in our application to improve the query time to constant while using exactly $nw$ bits, which we elaborate below.

The previous retrieval data structure is based on a linear mapping of the values.
We view the final retrieval data structure as a vector (of dimension $n$) in the field $\F_{2^w}$.
We will use free randomness to sample a random \emph{sparse} matrix $A$ of size $N\times n$ in a certain way, where each row corresponds to one possible key.
When the query algorithm gets a key $x$, it simply computes the inner product of row $x$ with the retrieval data structure (as a vector).
The time it takes to compute this inner product equals the number of non-zero entries in row $x$, and since $A$ is sparse, this computation is fast.
The main task is to show that there exists a vector as the final data structure, such that for all $n$ valid keys, the output is correct.
This requires us to solve a system of $n$ linear equations with $n$ unknowns, and there is a solution when the $n$ rows corresponding to the valid keys are linearly independent.
Hence, one way to achieve $O(\log n)$ recovery time, is to sample a random $A$ with $O(\log n)$ random non-zero entries per row, filled with random field elements in $\F_{2^w}$.
Then one can prove that any given set of $n$ rows has full rank with high probability.

To reduce the query time, we apply a novel trick to further sparsify the matrix.
We first observe that many $(i,j)$ pairs are always valid.
This is because each bucket has $B\pm B^{2/3}$ keys, hence, the size of each $D_i$ is between $\dmin := \frac{1}{w} \log \binom{V}{B-B^{2/3}}$ and $\dmax := \frac{1}{w} \log \binom{V}{B+B^{2/3}}$ words.
Therefore, all $(i,j)$ where $j<\dmin$ are always valid keys in the retrieval problem.
The only varying part based on the input is for $j$ between $\dmin$ and $\dmax$.
Observe that $\dmax-\dmin \ll \dmin$, hence, most of the valid keys are known to us in advance.
We will take advantage of those keys to sparsify the matrix.

This observation allows us to first allocate the same number of columns as the number of keys in the varying part, and only sample a random matrix in the rows corresponding to $j\in[\dmin,\dmax]$.
For the fixed valid inputs, we can put any fixed sparse full rank matrix, e.g., an identity matrix.
Note that the random part is a small fraction of the rows; hence, the total number of non-zero entries is not large compared to the total number of valid rows.
We then revise the matrix by evenly ``redistributing'' the non-zero entries in the random part to all rows, such that the revised matrix preserves the linear independence of the rows.
See \cref{fig:sparsify} for an example of the redistribution.

In \cref{sec:word_ram}, we show a different approach to construct a sparse invertible matrix without assuming free randomness.
The construction is much more sophisticated than sampling random non-zero entries.
It is based on hashing the rows into blocks in a hierarchical way, and it makes use of the Cauchy matrices.\footnote{A Cauchy matrix is a dense matrix, but any square submatrix of it has full rank.}
Also note that the random permutation can be replaced by one with bounded-wise independence.
Since these are the only two places where we use free randomness, this allows us to remove the assumption.

\subsection{Fractional-length data structures}
\label{sec:fractional_len}

We have assumed that each $D_i$ simultaneously has $1/\poly n$ redundancy and occupies an integer number of words.
This may not be true in general, since $\log\binom{V}{s_i}$ may be far from any integer multiple of $w$.
We use the standard solution---the \defn{spillover representation} of \Patrascu~\cite{patrascu2008succincter}.
That is, each data structure $D_i$ is represented by a pair $(m, k)$ such that $m\in\{0,1\}^M$ is a bit string and $k\in[K]$ is an integer, where $K = 2^{\Theta(w)}$; $k$ is called the \defn{spill} of the representation.
The query algorithm is assumed to be able to access either $k$ or any word in $m$ in constant time.
Intuitively, since there are $2^M\cdot K$ different such $(m, k)$ pairs, the data structure is defined to have length $M+\log K$.
Such a presentation gives a much more fine-grained measure of space, because increasing $K$ to $K+1$ only increases the length by $\log \frac{K+1}{K} = O(1/K) = 1/\poly n$ bits.
Yu~\cite{yu2020nearly} showed that for each bucket with $\poly w$ keys, a data structure with $1/\poly n$ redundancy in this representation can be constructed.
Furthermore, we may assume without loss of generality that $M$ is a multiple of $w$, since the remaining $< w$ bits can always be merged into $k$.
Moreover, we show that it is possible to ``embed'' the value of $s_i$ into the first word of $m$ such that the encoding of $s_i$ also incurs small redundancy.
Hence, the data structure $D_i$ we construct for each bucket is encoded in the spillover representation, which has a variable fractional length depending on $s_i$ with $1/\poly n$ bits of redundancy, and the value of $s_i$ is always encoded in the first word.

The complete-word parts $m$ are hence encoded using the retrieval data structure discussed in the previous subsection.
The spills $k$ may have different ranges $[K]$, so it is unclear if they can be encoded in a single retrieval data structure.
Observe that $s_i$ has at most $O(B)=\poly\log n$ many different values, so $K$ can also take at most $\poly\log n$ values (as its value is determined by $s_i$).
The spills $k$ can be stored in $\poly\log n$ retrieval data structures, each for one possible value of $K$.\footnote{This also requires us to change the underlying field from $\F_{2^w}$ to $\F_p$ for a $p$ larger than but very close to $K$.}
Given a query $x$, we first find the $D_i$ that needs to be queried, retrieve the first word of $D_i$ to recover $s_i$.
This determines the range $K$, and we get to know which retrieval data structure to query to recover the spill; then the query algorithm on $D_i$ can be simulated.
The details can be found in the following sections.

\section{Retrieval Data Structures}
\label{section:retrieval}
As discussed, no known general static retrieval data structure achieves $\poly \log n$ redundancy and constant query time simultaneously, so we cannot directly use it in our algorithm. However, we can exploit the additional property of our problem to further improve the retrieval data structure, namely the fact that we are required to also store a fixed-size array of size much larger than the retrieval.

We introduce a variant of retrieval data structures, which we call the \defn{augmented retrieval}: Besides the original functionality of retrieval data structure, we also need to store an array of elements $a_1, \ldots, a_m$, and we need to answer queries to the $i$-th element in $O(1)$ time (we will call the queries to the augmented array the \defn{augmented queries}). One straightforward way to implement the augmented retrieval is to store the retrieval data structure and the array separately. But we show that, by encoding these two parts of information together, we can get better space bounds.

Intuitively, our improvement works as follows: Similar to many previous constructions of retrieval data structures, in order to answer the retrieval queries, we first sample a sparse matrix $A$ of $N$ rows and $n$ columns, such that any $n\times n$ submatrix is row independent with good probability. Then we store a vector $\vec{b}$ of length $n$ in the memory, and each retrieval query is answered by computing the inner product of $\vec{b}$ and some row in $A$. A problem with this approach is that, in order for an $n\times n$ matrix to have full rank with good probability, we need at least $\Omega(\log n)$ non-zero entries per row, which means that each query requires $\Omega(\log n)$ time. To deal with this, we notice that the augmented queries can also be seen as computing an inner product, so we can put them in the same matrix as the retrieval queries. We can then use $O(\log n)$ augmented queries to ``sparsify'' each retrieval query, and make each row have only a constant number of non-zero entries. This sparsifying technique is the key technical contribution of this section, and it may be of independent interest.

Throughout this section, we will assume that both the values stored in the augmented retrieval (the ones associated to each key as well as the array elements) and the memory words are from a finite field $\F$ of order $2n\le|\F|=\poly n$. Also, the algorithm can access any memory word in constant time. Later, we will apply techniques from \cite{dodis2010changing} to simulate this type of memory on a word RAM.

\begin{lemma}
    \label{lem:augmented_retrieval}
    There is an augmented retrieval data structure that stores $n$ key-value pairs, where the keys are from a universe $[N]$ and values are elements in $\F$, as well as an augmented array $a_1, \ldots, a_m$ of elements in $\F$ where $m = \Theta(N \log n)$, such that:
    \begin{itemize}
    \item The data structure uses $n + m$ memory words in $\F$ and can answer retrieval queries and augmented queries in $O(1)$ worst-case time. In other words, the data structure introduces \emph{no} redundancy.
    \item The data structure assumes free access to a fully independent hash function from outside. Given such a random hash function, the data structure has a constant success probability in its construction process.
    \end{itemize}
\end{lemma}

We can boost the success probability to $1-1/\poly N$ by paying an extra $O(\log N)$ bit of redundancy, by using $O(\log N)$ independent hash functions, and storing the index of the hash function (if any) that makes the construction successful.

\begin{corollary}
    \label{thm:augmented_retrieval}
    In \cref{lem:augmented_retrieval}, we can boost the success probability to $1-1/\poly N$ for any $\poly N$, by paying an additional redundancy of $O(\log N)$ bits.
\end{corollary}

We first set up some notations used in the proof of \cref{lem:augmented_retrieval}. Let \textsc{Query}($i$) ($i\in [N+m]$) denote a query to the augmented retrieval, where we answer the value corresponding to key $i$ if $i\le N$ (or answer arbitrarily if $i$ is not present), and answer the $(i-N)$-th value in the array $a_{i-N}$ if $i>N$. We use $S\subset N$ to denote the set of keys that are present, and use $T \defeq S \cup (N, \, N + m]$ to denote the set of queries to the augmented retrieval that we need to answer correctly.

As mentioned before, our goal is to sample a $(N+m)\times (n+m)$ matrix $A$, with the property that $A_{T,*}$ is row independent with constant probability.\footnote{We use the notation $A_{I,J}$ to represent the submatrix of $A$ consisting of a set $I$ of rows and a set $J$ of columns. When $I$ or $J$ is replaced with ``$*$'', it represents the set of all rows or all columns.} Then in the preprocessing phase, we compute a vector $\vec{b}$ such that $A_{i,*}\cdot \vec{b}$ is the correct answer of \textsc{Query}($i$) for every $i\in T$, then store $\vec{b}$ in the memory. Each query can be answered by going over all the non-zero entries in $A_{i,*}$, and accessing the corresponding values of $\vec{b}$. Therefore, the time cost for answer \textsc{Query}(i) is the number of non-zero entries in $A_{i,*}$. So the question becomes constructing such a matrix $A$ such that it only has a constant number of non-zero entries per row.

Instead of directly constructing $A$, we first construct an intermediate matrix $B$ that is less sparse. $B$ is also a $(N+m)\times (n+m)$ matrix, and every row of $B$ will contain at most $t \defeq 10 \log n$ non-zero entries. $B$ is constructed as follows:

\begin{itemize}
    \item For each of the first $N$ rows, sample $t$ entries from the first $n$ columns with uniformly random positions and values, and let everywhere else be zero. Denote those entries by $\{(p_{i,j},v_{i,j})\}_{1\le j\le t}\in [n]\times \F$. It might be the case that the $t$ positions that we sampled are not distinct (say $p_{i,j}=p_{i,j'}$), in which case we just let $B_{i,p_{i,j}}$ be the sum of their corresponding values (formally speaking, $B_{i,*}=\sum_{j=1}^{k} v_{i,j} \, \vec{e}_{p_{i,j}}$).\footnote{We use $\vec{e}_k$ to denote the $k$-th unit vector where only the $k$-th component equals to $1$, and all other components equal to $0$.} Also note that the non-zero entries in the first $N$ rows only appear in the first $n$ columns.
    \item For the last $m$ rows, simply let $B_{i+N,i+n}=1$ for $i\in [m]$, and $B_{i,j}=0$ elsewhere. This corresponds to explicitly storing $a_i$ in the $(i+n)$-th memory word.
\end{itemize}

The process of generating $B_{S, [n]}$ can be thought of as randomly sampling $t$ entries from each row (with repetition), then letting their values be uniformly random. 
Thus, the linear independence of the rows of $B_{T,*}$ can be shown by the following lemma, whose proof is deferred to \cref{app:random_matrix_rank}. 

\begin{restatable}{lemma}{matrixRank}
    \label{lem:random_matrix_rank}
    Let $\F$ be a finite field of order $\ge 2n$. Let $M\in \F^{n\times n}$ be generated by randomly sampling $t\defeq 10\log n$ entries for every row (with repetition), and let their values be uniformly random. Then $M$ is row independent with constant probability.
\end{restatable}

Based on $B$, we can construct an augmented retrieval that uses $n+m$ words, and answers queries in $O(\log n)$ time: When answering \textsc{Query}($i$) for $i\in [N]$, the algorithm first finds all non-zero entries in row $i$ by computing the list $\{(p_{i,j},v_{i,j})\}_{1\le j\le t}$, then accesses the memory to learn the entry of $\vec{b}$ at every coordinate $p_{i,j}$.

So far, we have not utilized the augmented queries in any interesting way. To further achieve $O(1)$ time per query, we leverage the abundance of augmented queries to construct the \emph{sparse} matrix $A$ with the desired row-independent property.

The matrix $A$ is constructed by performing a sequence of elementary operations over the matrix $B$ to sparsify its rows. Specifically, let $t \defeq m/N$. We uniformly partition all $m$ augmented rows (i.e., the rows corresponding to augmented queries) into $N$ groups, where the $i$-th group $G_i$ consists of the $t$ augmented queries in $(N + (i-1)t,\, N + it]$. 
Moreover, we will put the $i$-th row, which corresponds to a retrieval query $\textsc{Query}(i)$, into $G_i$, and use other rows in $G_i$ as a resource to sparsify the $i$-th row.

We first focus on the construction of $A$ within each group $G_i$. Let $A^{(i)}$ denote the $(t+1) \times (n+t)$ submatrix of $A$, where the rows are from $G_i$ and the columns are from $[n] \cup (n+(i-1)t, \, n+it]$. Initially, we set $A = B$, so the matrix $A^{(i)}$ looks like the left matrix of \cref{fig:sparsify}.

We perform two sets of operations on $A^{(i)}$:

\begin{itemize}
    \item Subtract row $j$ from row $j-1$, for every $j\in [2,\,t+1]$.
    \item Then for each $j\in [t]$, add $v_{i,j}$ times columns $n+1,n+2,\dots,n+j$ to column $p_{i,j}$ (recall that $(p_{i,j},v_{i,j})$ is the $j$-th non-zero entry on row $i$), which amounts to moving $v_{i,j}$ from the first row to the $(j+1)$-th row.
\end{itemize}

\begin{figure}[ht]
    \centering
    \begin{align*}
        \begin{pmatrix}
            v_{i,1}  & v_{i,2} & v_{i,3} & & & \\
            & & & 1 & & \\
            & & & & 1 & \\
            & & & & & 1 \\
        \end{pmatrix}\Rightarrow
        \begin{pmatrix}
            v_{i,1}  & v_{i,2} & v_{i,3} & -1 & & \\
            & & & 1 & -1 & \\
            & & & & 1 & -1 \\
            & & & & & 1 \\
        \end{pmatrix}\Rightarrow
        \begin{pmatrix}
            & & & -1 & & \\
            v_{i,1} & & & 1 & -1 & \\
            & v_{i,2} & & & 1 & -1 \\
            & & v_{i,3} & & & 1 \\
        \end{pmatrix}
    \end{align*}
    \caption{Sparsifying a row}
    \label{fig:sparsify}
\end{figure}

By combining the construction of each $A^{(i)}$ (and setting other undefined entries to zero), we obtain the matrix $A$, where each row in $A$ contains at most $3$ non-zero entries. Furthermore, for any $i\in [N+m]$, we can easily compute the non-zero entries in $A_{i,*}$ in constant time, allowing the algorithm to answer each query in constant time. Finally, since $A$ is obtained by performing some elementary operations on $B$, where we only subtract rows in $(N, \, N + m]$ (which are guaranteed to be in $T$), we can conclude that $A_{T, *}$ is row independent with constant probability, just as $B_{T,*}$ is, thereby justifying our construction.

\section{Cell-Probe Dictionary}
\label{sec:cell_probe}
In this section, we present a simplified version of our membership data structure with $\poly \log n$ bits of redundancy that works in the cell-probe model.
This data structure demonstrates our main ideas, and it uses the retrieval data structure from Section~\ref{section:retrieval}.
In Section~\ref{sec:word_ram}, we will show how to replace its components so that it works in the word RAM model.

\begin{theorem}
    \label{thm:cell_probe_main}
    In the cell-probe model, there is a static dictionary storing $n$ keys from a universe $U \in [2n, \poly n]$, using $\log \binom{U}{n} + \poly \log n$ bits of space, and answers membership queries in worst-case constant time.
\end{theorem}

We first assume the algorithm has access to a random tape, which means that it can evaluate fully random hash functions in constant time. We will remove this assumption in \cref{subsection:altogether}.

We partition all $n$ keys into a number of \defn{buckets}, where each bucket contains $B = \poly \log n$ keys in expectation, and we use $s_i$ to denote the actual number of keys hashed to the $i$-th bucket. With high probability in $n$, every bucket will contain $[B - B^{2/3}, \, B + B^{2/3}]$ keys.

\subsection{Encoding of a bucket}\label{subsection:intra_bucket}

We first show how to encode the information within each bucket. Let $V\defeq U/(n/B)$ denote the universe size of a bucket. Each bucket is responsible for answering queries to a small dictionary, where the keys are from $[V]$, and the number of keys is between $B - B^{2/3}$ and $B + B^{2/3}$ with high probability. Since $B = \poly \log n$, we can use techniques from \cite{patrascu2008succincter,yu2020nearly} to construct a data structure that has $1 / \poly n$ bits of redundancy, and answers queries to the small dictionary in constant time. We can obtain an encoding of such a data structure that uses $\log\frac{1}{p(s_i)} + \log \binom{V}{s_i}$ bits for bucket $i$, where $p$ as an arbitrary distribution over $[B - B^{2/3}, \, B + B^{2/3}]$. For now, readers should think of $p(s_i)$ as the probability that the $i$-th bucket gets $s_i$ keys. Later in \cref{subsection:altogether}, we will specify our hashing method and formally define $p$.

As mentioned in \cref{sec:fractional_len}, since we only incur $1 / \poly n$ bits of redundancy, the length of our data structure (for each bucket) may not be an integer number of bits. 
Formally, our data structure is encoded using a spillover representation $(m, k) \in \{0, 1\}^M \times \{0, \dots, K-1\}$, such that $K=\poly n$. When stored optimally (as we will do in \cref{subsection:concat}), this representation uses space $M+\log K$ bits, which can be fractional.

In the remainder of this subsection, we fix a specific bucket, and use $S\subset [V]$ to denote the set of keys hashed to this bucket, where $s\defeq |S|\in [B - B^{2/3}, \, B + B^{2/3}]$. 

\begin{lemma}
  \label{lem:bucket_rep}
  We can construct a spillover representation for $S$, which is a tuple $(m,k)\in \{0,1\}^M\times \{0, 1, \dots, K \!-\! 1\}$, where $M, K$ are positive integers, $K = \poly n$, and 
  \begin{align*}
    M+\log K\le \min\bk*{\log\frac{1}{p(s)}, \, O(\log n)}+\log \binom{V}{s}+O(1/n^2).
  \end{align*}
  The representation $(m,k)$ has the following properties that help answer queries:
  \begin{itemize}
    \item The first $O(1)$ words of $m$ determine the set size $s$.
    \item The size parameters $M,K$ are uniquely determined by $s$.
    \item After learning the set size and the value of the spill $k$, answering whether a key is in $S$ can be done with $O(1)$ word probes to $m$ (and $O(1)$ time) assuming access to a lookup table of size $O(n^{\epsilon})$, where $\epsilon$ can be chosen to be any positive number. This lookup table does not depend on $S$.
  \end{itemize}
\end{lemma}

We remark that the membership queries in \cref{lem:bucket_rep} are efficient not only in the cell-probe model but also in the word RAM model, assuming access to the lookup table of $O(n^{\eps})$ space. 
Although this section focuses on the cell-probe model, the time efficiency and lookup tables will be important for later sections, when we switch to the word RAM model.
For now, we will not count the lookup table toward our space usage, because in the cell-probe model, any computation that does not depend on the memory is free.

The proof of \cref{lem:bucket_rep} is a combination of existing techniques. First, we apply \cite[Lemma 28]{yu2020nearly} to obtain a spillover representation $(m', k') \in \{0, 1\}^{M'} \times \{0, \dots, K' \!-\! 1\}$ for $S$ conditioned on the set size $s$, which has the following properties:

\begin{itemize}
  \item $K' = \poly n$,
  \item $M' + \log K' \le \log \binom{V}{s} + O(1/n^2)$,
  \item there is a lookup table of size $O(n^\epsilon)$,
  \item each query can be answered with $O(1)$ word probes to $m'$, $k'$ and the lookup table, assuming that we know $s$ in advance.
\end{itemize}

Given $(m',k')$, it remains to encode $s$ into the first $O(\log n)$ bits of $m'$. Here we can assume WLOG that $p(s) \ge 1/n^3$ for every $s$ by tweaking the distribution as in \cite{patrascu2008succincter}, which introduces $O(1/n^2)$ bits of redundancy. This is the reason for the $\min\bk[\big]{\log\frac{1}{p(s)}, \, O(\log n)}$ term in the lemma.

In order to encode $s$ into the first $O(\log n)$ bits of $m'$, we perform the following steps.
Let $t \defeq 10 \log n$. We partition $T = \{0, \dots, 2^{t} \!-\! 1\}$ into sets $\{T_s\}$, such that $|T_s|/|T|$ and $p(s)$ differ by at most $1/n^6$. Then we take the first $2t$ bits of $m'$, and let $m_0 \in [0, 2^{2t})$ denote its value.
Finally, we re-encode $m_0$: We replace the first $2t$ bits of $m'$ with the $(m_0 \bmod |T_s|)$-th element in $T_s$ (which is a $t$-bit string), shortening $m'$ to form $m$; we also concatenate $k'$ with $\lfloor m_0/|T_s|\rfloor$ to form the new spill $k$.

After modification, the length of $m'$ shrinks by $t$ bits, and the spill universe becomes $K' \cdot \lceil 2^{2t}/|T_s| \rceil$, so its logarithm increases by 
\[\log \, \lceil 2^{2t}/|T_s| \rceil \le \log\bk*{\frac{2^t}{p(s)} \cdot \bk[\big]{1+O(1/n^3)} \cdot \bk[\Big]{1 + \frac{1}{2^t}}} = \log \frac{1}{p(s)} + t + O(1/n^2),\]
where the first inequality holds because $|T_s| \ge |T| \cdot (p(s) - 1/n^{6}) \ge |T| \cdot p(s) \cdot (1 - 1/n^{3})$.
Overall, the new representation only incurs $O(1/n^2)$ redundancy, as claimed.

This way, we can guarantee that the first $t$ bits in $m'$ determine $s$ (to obtain $s$, we only need to check which $T_s$ it is in), and that we can recover $m_0$ by reading the first $t$ bits of $m'$ and the spill.

Finally, we discuss how to recover $s$ and $k'$ time-efficiently. For this, we need to read the first $t$ bits of $m'$, then decide which set $T_s$ it is in, as well as how many elements in $T_s$ are smaller than it. To do this, we can let $T_s$ be continuous intervals, and order the intervals by their lengths (e.g., the smallest $T_s$ occupies the smallest $|T_s|$ elements in $T$, and so on). Then we can use the predecessor search technique in \cite{patrascu2008succincter} to recover $s$ in constant time, using a data structure of size $O(B)$ words.

\subsection{Concatenate the representations of each bucket}
\label{subsection:concat}

The main challenge of the problem is to concatenate all the spillover representations $(m_i,k_i)$ of each bucket, while supporting constant query time in the worst case. Our algorithm depends on the following observations.

\begin{itemize} 
\item The number of elements $s_i$ in each bucket is well concentrated near $B$, and the length of $m_i$ is therefore also concentrated near some length, therefore we can view $m_i$ as the combination of a long fixed-length part and a short variable-length part.
\item The length of $m_i$ and the size of the spill universe for $k_i$ depend only on $s_i$, so there are few types of different sizes. We will use $(M^{(s)}, K^{(s)})$ to denote the size parameters $(M_i, K_i)$ of the spillover representation of any bucket $i$ with size $s$.
\end{itemize}

Furthermore, we assume that $m_i$ is a sequence of complete words, i.e., its length is a multiple of $w$. If this is not the case, we will truncate $m_i$ to complete words, and merge the leftover $O(w)$ bits into the spill.

The problem of concatenating all spillover representations $(m_i, k_i)$ with their sizes $s_i$ is addressed by the following lemma.

\begin{lemma}\label{lem:storing_representations}
Given $L$ spillover representations $(m_i, k_i)$ and size parameters $s_i$ that satisfy the following:
\begin{itemize}
\item There are $S$ different values of $s_i$.
\item $m_i$ is a variable-length sequence of words. The length only depends on $s_i$ and is in the range $[M_{\min},\, M_{\max}]$, that is, $|m_i|=M^{(s_i)}\in [M_{\min},\,M_{\max}]$.
\item $s_i$ is fully determined by the first $O(1)$ words of $m_i$.
\item $k_i$ is an integer in $[K_i]$, and the size of the universe also only depends on $s_i$, denoted as $K_i = K^{(s_i)}$. We also require that $L^3<K_i =\poly L$.
\item The parameters satisfy
\[
M_{\min} \ge \Omega\bk*{(M_{\max} - M_{\min}) \log (L M_{\max})+S\log L}. \numberthis\label{inequality_store_representations}
\]

\end{itemize}

In the cell-probe model with word size $w = \Theta(\log n)$, there is a data structure that stores all tuples $(m_i, k_i)$, supporting queries to any word in $m_i$ or the entire spill $k_i$ in constant time in the worst case.

The data structure incurs $O(S\log L)$ words of redundancy, assuming access to a fully independent hash function, and its construction succeeds with probability $1-1/\poly L$ over the randomness of the hash function.
\end{lemma}

For the ease of discussion, we assume that the data structure knows $s_i$ when handling a query of the $i$-th spillover representation. We can ensure this by storing the first $O(1)$ words of each $m_i$ separately, and removing them from $m_i$. This only causes $M_{\min}, M_{\max}$ to shrink by $O(1)$ words, which does not make any difference to the algorithm.

Next, we present our data structure for \cref{lem:storing_representations}. Our data structure breaks the spillover representations into three parts.

\begin{itemize}
    \item The fixed-length part, which consists of the first $M_{\min}$ words of each $m_i$.
    \item The variable-length part, which consists of the remaining part of each $m_i$, where the length of each remaining part does not exceed $M_{\max}-M_{\min}$ words.
    \item The spills $k_i$.
\end{itemize}

For simplicity, our data structure explicitly stores the parameters $S, M_{\min}, M_{\max}, M^{(s)}, K^{(s)}$ and each type of $s_i$, causing only $O(S)$ words of redundancy.
Storing the fixed-length part is simple: Since we take a fixed number of words from each $m_i$, these $L M_{\min}$ words can be stored sequentially. We denote this sequence by $\mfix$. 

In the following paragraphs, we address the variable-length parts and the spills.

\paragraph{Storing variable-length parts.}

We first recall the augmented retrieval problem defined in \cref{section:retrieval}. In the retrieval part, the data structure needs to support some key-value queries from a key universe $[N]$ and a value field $\F$. To support these queries, the data structure also stores an augmented array in $\F$ of length at least $\Theta(N \log N)$. It can achieve constant query time and $O(1)$ words of redundancy with the help of the augmented array.

The task of storing variable-length parts can be modeled as a retrieval problem:
The $j$-th word in the $i$-th variable-length part is viewed as a key $(i,j) \in [L] \times [M_{\max} \!-\! M_{\min}]$ in the retrieval problem, with the associated value being the content of that word.
When we need to access some word in the variable-length part (assuming that we already know $s_i$ and thus $M_i$), we can just query the retrieval data structure with the corresponding key and get the result in constant time.

We now use \cref{thm:augmented_retrieval} to solve this retrieval problem, using part of $\mfix$ as the augmented array. The length of the augmented array required by the lemma is $\Theta(N \log N)$ words, so we will take the first
\[
\Theta\bk*{L(M_{\max}-M_{\min})\log\bk*{L(M_{\max}-M_{\min})}}
\]
words of $\mfix$ to serve as the augmented array, and combine them with retrieval queries in $[L] \times [M_{\max}-M_{\min}]$ to form an augmented retrieval data structure. The field $\F$ in the augmented retrieval is set to be $\GF(2^w)$, the finite field with size $2^w$. (In the cell-probe model, arithmetic operations over the field $\GF(2^w)$ can be completed in constant time.)

\paragraph{Storing spills.}

Recall that in the retrieval problem, all values are from the same field $\F$. The task of storing spills is different from this setting since each spill may come from a different range $[K_i]$.
Nevertheless, there are only $S$ types of different universes $[K^{(s)}]$, so we can build data structures for each type of universe separately. Formally, we say that the spill $k_i$ has type $s_i$, and we will split the problem into $S$ retrieval problems, each storing all spills of a certain type $s$ (which are in the same range). Our algorithm for storing spills is as follows: For each different $s$, we consider all spills of type $s$, and construct an augmented retrieval data structure with keys being the indices of these spills and the value being the content of the spill (which is in $[K^{(s)}]$). Similar to before, the augmented elements required by the augmented retrieval are taken from $\mfix$.

However, the above solution has two issues. First, \cref{thm:augmented_retrieval} requires that the values are from a field, but $K^{(s)}$ is not necessarily a prime power. We apply the following fix: For each possible $s$, we find the smallest prime larger than $K^{(s)}$, denoted by $P^{(s)}$. For the augmented retrieval for type $s$, the data structure uses $\F_{P^{(s)}}$ as the field. The following lemma bounds the gap from $K^{(s)}$ to the next prime $P^{(s)}$, and thus the redundancy introduced by this fix is small.

\begin{lemma}[\cite{NextPrime}]\label{lem:nextprime}
    For sufficiently large $n$, there exists a prime in $[n, \, n+n^{7/11}]$.
\end{lemma}

Therefore, the redundancy introduced for a single spill is
\[\log P^{(s)}-\log K^{(s)}=\log\frac{P^{(s)}}{K^{(s)}}\le \log \bk*{1+\bk[\big]{K^{(s)}}^{-4/11}}\le O\bk*{\bk[\big]{K^{(s)}}^{-4/11}}\]
bits,
so this fix introduces less than $O(1/L)$ bits of redundancy overall since $K^{(s)}>L^3$.

The second issue is more critical: The augmented retrieval data structure under $\F_{P^{(s)}}$ requires an augmented array in field $\F_{P^{(s)}}$, but $\mfix$ is stored in machine words of base $2^w$. To build the augmented array and use \cref{thm:augmented_retrieval}, we need to convert an array from base $2^w$ to base $P^{(s)}$. Moreover, the data structure constructed by \cref{thm:augmented_retrieval} is stored as a sequence of elements in $\F_{P^{(s)}}$. To store this data structure in memory words, we need to convert its representation from base $P^{(s)}$ back to base $2^w$ again. To perform base conversions, we use the following result from \cite{dodis2010changing}.%
\footnote{The original theorem only states how to convert from arbitrary base $b$ to binary words of base $2^w$, but their proof can be naturally generalized to arbitrary target base $q$.}

\begin{theorem}[Implicit in {\cite{dodis2010changing}}]\label{thm:changing_base}
In the word RAM model, given $n,p,q$ where $p=\poly q$, $n\le \poly q$, we can represent a sequence $A[1\ldots n]$ of elements in $[p]$ using a sequence $B$ of elements in $[q]$, such that
\begin{itemize}
    \item $B$ is of length $n \log_p q + O(\log q)$. This implies that conversion only incurs $O(\log^2 q)$ bits of redundancy.
    
    \item Each element in $A$ can be restored in constant time by accessing elements in $B$ under the word RAM model with word size $\Omega(\log q)$.
\end{itemize}
\end{theorem}

Using this theorem, we can construct the retrieval data structure for each type $s$ as follows: First, we convert part of $\mfix$ into elements in $[P^{(s)}]$, use them to build the augmented retrieval under $\F_{P^{(s)}}$, and finally convert the retrieval data structure back to word representation. In particular, our algorithm will perform the following steps for each $s$. (Also see \cref{fig:storing_spills}.)

\begin{enumerate}
    \item\label{item:step1} Take the next $\Theta(L \log L)$ words in $\mfix$ that are not used in the previous steps of the construction, and convert this sequence of words to base $P^{(s)}$ using \cref{thm:changing_base}. Let the resulting sequence be $\mconv$, which has $\Theta\bk[\big]{L \log L \cdot \frac{w}{\log P^{(s)}}} = \Theta(L \log L)$ elements since $K^{(s)} = \poly n$ and $P^{(s)}$ is the next prime after $K^{(s)}$ (and thus $P^{(s)} \le 2K^{(s)}$).
    \item Construct an augmented retrieval data structure as discussed earlier: The keys of the retrieval are indices $i$ of spills of type $s$, which are from the key universe $[L]$, and the associated values are the spills $k_i$. The augmented retrieval (\cref{thm:augmented_retrieval}) requires an augmented array of length $\Theta(L \log L)$, for which we use the sequence $\mconv$ from the previous step. The resulting augmented retrieval data structure can be represented using a sequence of elements in $\F_{P^{(s)}}$, denoted by $\mretr$.
    \item Convert $\mretr$ back to a sequence of words (in base $2^w$) using \cref{thm:changing_base} again, then store these words in the memory sequentially.
\end{enumerate}

\begin{figure}[ht]
    \centering
    \begin{tikzpicture}
    \def \leftspace {1};
    \def \stepheight {-2.5};
    \def \mytick {0.05};
    \def \drawtick[#1,#2]{\draw ({#1}, {#2-\mytick}) -- ++ (0,2*\mytick);} %
    \def \drawlargetick[#1,#2]{\draw ({#1}, {#2-2*\mytick}) -- ++ (0,4*\mytick);}
    \def \drawsegment[#1,#2,#3,#4,#5,#6,#7,#8]{ %
        \draw ({(#1)}, {(#2)}) -- ({(#1) + (#3)}, {(#2)})
         node[midway, above = 0.1] (#7) {#5}
         node[midway, below = 0.1] (#8) {#6};
        \drawlargetick[(#1),(#2)]
        \drawlargetick[(#1) + (#3),(#2)]
        \foreach \i in {1,...,#4}
        {
            \drawtick[(#1) + \i / (#4 + 1) * (#3),(#2)]
        }
    }
    \drawsegment[\leftspace,0,5,9,$\mfix$,$[2^w]$,,lower_mfix]
    \begin{scope}[local bounding box=subqueries]
        \drawsegment[\leftspace + 12/2,0,0.75,0,,,,]
        \drawsegment[\leftspace + 14/2,0,0.75,0,,,,]
        \drawsegment[\leftspace + 17/2,0,0.75,0,,,,]
        \node[] at (\leftspace + 14/2 + 9/8, 0) () {$\cdots$};
    \end{scope}
    \node[below=-0.25 of subqueries] (lower_spill) {$[K^{(s)}]$};
    \node[above=-0.25 of subqueries] {$k_i$};
    \drawsegment[\leftspace,\stepheight,5,4,$\mconv$,$[P^{(s)}]$,upper_mconv,lower_mconv]
    \begin{scope}[local bounding box=subqueries2]
        \drawsegment[\leftspace + 12/2,\stepheight,0.75,0,,,,]
        \drawsegment[\leftspace + 14/2,\stepheight,0.75,0,,,,]
        \drawsegment[\leftspace + 17/2,\stepheight,0.75,0,,,,]
        \node[] at (\leftspace + 14/2 + 9/8, \stepheight) () {$\cdots$};
    \end{scope}
    \node[below=-0.25 of subqueries2] (lower_spill2) {$[P^{(s)}]$};
    \node[above=-0.25 of subqueries2] (upper_spill2) {};
    \node[] at (0, \stepheight * 0.5) () {Step 1:};
    \draw [->, color=red] (lower_mfix.south) -- (upper_mconv.north) node[midway,left] {Changing Base};
    \draw [->, color=red] (lower_spill.south) -- (upper_spill2) node[midway,left] {};
    \drawsegment[\leftspace,2*\stepheight,9,8,$\mretr$,$[P^{(s)}]$,upper_mretr,lower_mretr];
    \node[] at (0, \stepheight * 1.5) (tmp) {Step 2:};
    \draw [->, color=red] (lower_mconv.south) -- (upper_mretr.north west);
    \draw [->, color=red] (lower_spill2.south) -- (upper_mretr.north east);
    \node[color = red] at (\leftspace + 9.5/2, \stepheight * 1.4) (tmp) {Augmented Retrieval};
    \drawsegment[\leftspace,3*\stepheight,9,17,,$[2^w]$,upper_final,lower_final];
    \node[] at (0, \stepheight * 2.5) (tmp) {Step 3:};
    \draw [->, color=red] (lower_mretr.south) -- (upper_final) node[midway,left] {Changing Base};
\end{tikzpicture}
    \caption{Storing all spills of type $s$. (1) First, we convert part of $\mfix$ from base $2^w$ to base $P^{(s)}$ using \cref{thm:changing_base}, getting $\mconv$; we also ``round up'' the universe of spills from $[K^{(s)}]$ to $[P^{(s)}]$. (2) Next, we build the augmented retrieval by \cref{thm:augmented_retrieval} using $\mconv$ as augmented elements, getting the representation of this data structure $\mretr$. (3) Finally, we convert it back to binary words using \cref{thm:changing_base}.}
    \label{fig:storing_spills}
\end{figure}
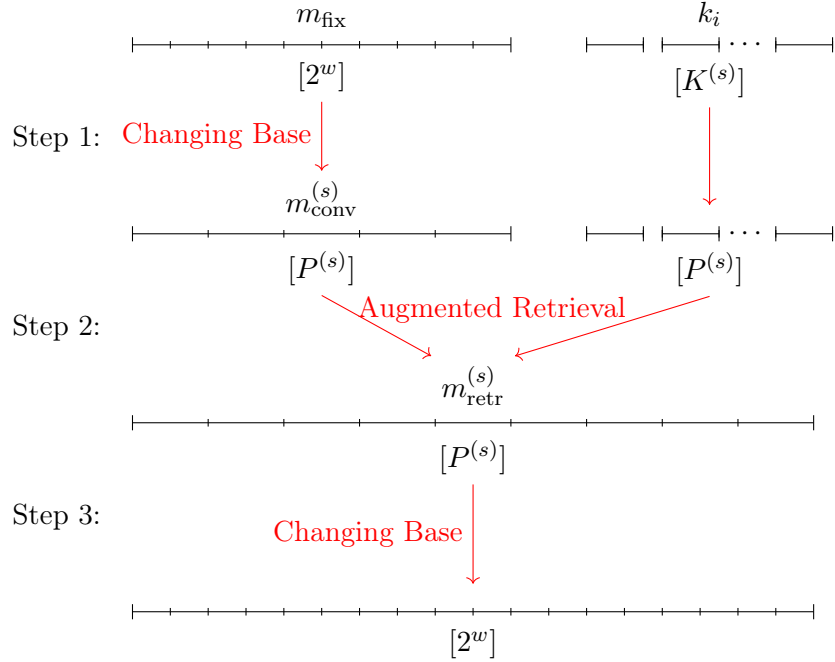

Next, we illustrate the process of answering queries for spillover representations. There are two types of queries involved in this construction.

\begin{itemize}
\item Querying a spill $k_i$ of type $s$. The algorithm queries the augmented retrieval data structure above to get $k_i$, which leads to a constant number of probes on $\mretr$. By \cref{thm:changing_base}, each probe on $\mretr$ takes a constant time accessing the binary words stored in the memory.
\item Querying a word in the part of $\mfix$ used in the above construction. By the base-conversion process in Step~\ref{item:step1}, this word can be found by probing a constant number of elements in $\mconv$. Each probe of $\mconv$ is a query to the augmented retrieval data structure, which takes $O(1)$ probes to $\mretr$, each of which takes $O(1)$ accesses to the memory words. Thus, the entire query process still takes a constant time to complete.
\end{itemize}
In summary, both types of queries take constant time in the worst case.

\paragraph{Combining two parts.}

We have discussed how to concatenate the variable-length parts of $m_i$ and the spills $k_i$. The final construction for \cref{lem:storing_representations} is a simple combination of these two parts: We build an augmented retrieval that stores the variable-length parts of $m_i$; we also build an augmented retrieval for each type $s$ that stores all spills of type $s$. Within each of the augmented retrievals, several words in $\mfix$ are also encoded as augmented elements. We store these augmented retrievals sequentially in memory one by one, together with a pointer to the beginning of each augmented retrieval data structure.

To allow queries to the words in $\mfix$, we additionally store two numbers: the number of words in $\mfix$ that are used in the augmented retrieval for the variable-length parts, and that for spills---for different types of spills $s$, we use the same number of words from $\mfix$ in the augmented recovery. Given these two numbers, when the user queries a word in $\mfix$, the data structure knows which augmented retrieval contains the queried word (if there is one), and the index of the queried word within that augmented retrieval. The query can then be completed in constant time.

\paragraph{Correctness analysis.} In the above paragraphs, we already show that the queries are correct and take constant time. Now we show the remaining parts of the correctness. We first show that when constructing the augmented retrieval data structures, we have enough words of $\mfix$ to build the augmented array. When storing the variable-length part, it uses $\Theta\bk*{L(M_{\max}-M_{\min})\log\bk*{L(M_{\max}-M_{\min})}}$ words as the augmented array. When storing the spills, for each possible $s$, it converts $\Theta(L\log L)$ words to get the augmented array under base $P^{(s)}$. The total number of used words is thus bounded by
\begin{align*}
    &\Theta\bk[\big]{L(M_{\max}-M_{\min})\log\bk*{L(M_{\max}-M_{\min})}+SL\log L}\\
    \leq{} & L \cdot \Theta\bk[\big]{(M_{\max}-M_{\min})\log(LM_{\max}) + S\log L}\\
    <{} & L \cdot M_{\min},
\end{align*}
where the last inequality holds according to the requirement \eqref{inequality_store_representations} of the lemma. So, the number of words in $\mfix$ is more than we need.

As for the success probability, recall that we only construct $S+1$ augmented retrievals. For each of them, there is a $1/\poly L$ probability to fail, where the $\poly L$ factor can be any polynomial in $L$, as stated in \cref{thm:augmented_retrieval}. Therefore the overall success probability is $1-1/\poly L$.

Finally, we consider the redundancy of this construction. The redundancy consists of the following parts:

\begin{itemize}
    \item The data structure first uses $O(S)$ words to store auxiliary information, including the pointers to each augmented retrievals, $O(1)$ extra parameters, and the number of words from $\mfix$ that are used in each augmented retrievals.
    \item To make the construction of the augmented retrieval succeed with high probability, we need to sample $O(\log L)$ hash functions and store one index with $O(\log \log L)$ bits, and we store it in a machine word. The data structure constructs $S + 1$ augmented retrievals, so these indices incur $O(S)$ words of redundancy.
    \item By \cref{thm:changing_base}, when the data structure stores each type of spill, each base-conversion step incurs $O(\log L)$ words of redundancy. These costs sum up to $O(S \log L)$ words of redundancy.
    \item Finally, rouding up the spill size from $K^{(s)}$ to $P^{(s)}$ results in $\log \frac{P^{(s)}}{K^{(s)}}$ bits of redundancy for each spill. By \cref{lem:nextprime}, this part is less than
    \[
    \log\bk*{1+\bk[\big]{K^{(s)}}^{-4/11}}<O\bk*{\bk[\big]{K^{(s)}}^{-4/11}}
    \]
    bits. Since $K^{(s)}\geq L^3$, this is less than $O(L^{-1})$ bits for each spill, which sums up to $O(1)$ bits for all spills.
\end{itemize}
Therefore, the total redundancy of our data structure is $O(S \log L)$ words.

\subsection{Putting everything together}

In this subsection, we combine the subroutines from previous subsections and prove \cref{thm:cell_probe_main}.

\begin{proof}[Proof of \cref{thm:cell_probe_main}]

We first partition all $n$ keys into $n/B$ buckets using a hash function, where $B=\poly \log n$. Our hash function is given by the following lemma, whose proof is deferred to \cref{app:random_permutation}.

\begin{restatable}{lemma}{randomPermutation}
    \label{lem:random_permutation}
    Let $B\ge \log ^4 n$, and let $\epsilon = \Omega(\sqrt{\log\log U/\log B})$. There is a family $\mathcal{H}$ of random permutations $h : [U] \to [U]$, with the following guarantees:
    \begin{itemize}
        \item Partition $U$ into $n/B$ buckets, each of size $U/(n/B)$, and fix a set $S$ of $n$ keys in $U$. Then with probability $\ge 1-1/(4n^2)$, the number of keys hashed to each bucket is in $[B-B^{2/3}, \, B+B^{2/3}]$.
        \item A member in $\mathcal{H}$ can be described using $O(n^\epsilon)$ bits, and can be evaluated in constant time.
    \end{itemize}
\end{restatable}

We remark that later in the word RAM model, we will require that the overall redundancy is $O(n^\eps)$ for any constant $\eps > 0$. In that case, we can set $B$ to be a sufficiently large power of $\log n$ in order to reduce the cost of storing the hash function.

Fixing the set of keys that the dictionary needs to store and taking a random hash function, we define $p(s)$ as the probability that a randomly chosen bucket has size $s$. Since $p(s)$ may depend on the input key set, we need to explicitly store $p(s)$ for $s \in [B-B^{2/3}, \, B+B^{2/3}]$ when constructing the data structure.

Within each bucket, the data structure is responsible for storing a set of $s_i$ elements in $[V]$, so we use \cref{lem:bucket_rep} to obtain a spillover representation of length $\log \frac{1}{p(s_i)} + \log \binom{V}{s_i} + O(1/n^2)$. The following lemma shows that the total length of all spillover representations is at most $\log \binom{U}{n} + O(\log n)$.

\begin{lemma}
    \label{lem:bucket_entropy}
    Let $p(s)$ be defined as above. Then, we have
    \begin{align*}
        \E_{h \in \mathcal{H}}\Bk*{\sum_{i=1}^{n/B}\bk*{\log\frac{1}{p(s_i)}+\log \binom{V}{s_i}}}\le \log \binom{U}{n} + O(\log n).
        \numberthis\label{eq:bucket_entropy}
    \end{align*}
    where the expectation is taken over the choice of the hash function $h \in \mathcal{H}$ described in \cref{lem:random_permutation}.
\end{lemma}

\newcommand{\Dhash}{\mathcal{D}_{\textup{hash}}}
\newcommand{\Dunif}{\mathcal{D}_{\textup{unif}}}
\newcommand{\Dind}{\mathcal{D}_{\textup{ind}}}

\begin{proof}
    We prove by comparing three distributions of key sets $S$ over the universe $[U] = [n/B] \times [V]$:
    \begin{enumerate}
    \item The first distribution $\Dhash$ is defined by the following random process. For each bucket $i$, we \emph{independently} sample a number $s_i$ following the probability distribution $p(s)$, then uniformly sample an $s_i$-element subset of the range $[V]$ of that bucket. The entropy of this distribution is
    \[
    H(\Dhash) = \frac{n}{B} \cdot \E_{s \sim p}\Bk*{\log \frac{1}{p(s)} + \log \binom{V}{s}} = \E_{i \in [n/B], \, h \in \mathcal{H}}\Bk*{\log \frac{1}{p(s_i)} + \log \binom{V}{s_i}} = \textup{LHS of \eqref{eq:bucket_entropy}},
    \]
    where we slightly abuse the notation and use $p$ to indicate the distribution defined by $p(s)$; the last equality holds by the linearity of expectation.
    \item The second distribution $\Dind$ is defined by including every $x \in U$ in the key set $S$ with probability $n/U$ independently.
    \item The third distribution $\Dunif$ is defined by choosing an $n$-element subset of $U$ uniformly at random among all $\binom{U}{n}$ choices.
    \end{enumerate}
    
    We will show $H(\Dhash) \le H(\Dind) \le H(\Dunif) + O(\log n)$ to prove the lemma.

    To see $H(\Dhash) \le H(\Dind)$, we focus on the indicator random variables $\ind[x \in S]$ for each element $x \in U$. When $S \sim \Dind$, these indicators are independent variables each equals 1 with probability $n/U$; when $S \sim \Dhash$, each indecator follows the same marginal distribution, but different indicators might be correlated. The statement $H(\Dhash) \le H(\Dind)$ follows because the independent distribution maximizes entropy given the marginals.

    The second inequality $H(\Dind) \le H(\Dunif) + O(\log n)$ follows by a direct calculation:
    \begin{align*}
        H(\Dind) ={} &-U\bk*{\frac{n}{U}\log \frac{n}{U}+\bk*{1-\frac{n}{U}}\log (1-\frac{n}{U})} \\
        ={}& n \log \frac{U}{n} + \bk[\big]{U - n} \log \frac{U}{U - n} \\
        ={}& U\log U-n\log n-(U-n)\log (U-n);
    \end{align*}
    on the other hand, using Stirling's approximation $\log (n!)=n\log n-n\log e\pm O(\log n)$, we can show that
    \begin{align*}
        H(\Dunif) = \log \binom{U}{n} = U\log U-n\log n-(U-n)\log (U-n) \pm O(\log n),
    \end{align*}
    thus $H(\Dind) \le H(\Dunif) + O(\log n)$, which concludes the proof.
\end{proof}

\cref{lem:bucket_entropy} only showed that the expected space used by our data structure is bounded, but we prefer a worst-case space bound. Since the LHS of \eqref{eq:bucket_entropy} is nonnegative, we can use Markov's inequality to show that 
\[
\Pr_{h \in \mathcal{H}}\Bk*{\sum_{i=1}^{n/B}\bk*{\log\frac{1}{p(s_i)}+\log \binom{V}{s_i}}\ge \log \binom{U}{n}+\log ^2n}\le \frac{\log \binom{U}{n}+O(\log n)}{\log \binom{U}{n}+\log ^2n}=1-\Omega\bk*{\frac 1n}.
\]
Therefore, with probability $\Omega(1/n)$, the data structure uses at most $\log \binom{U}{n} + \poly \log n$ bits of space. Then, we can sample $\poly n$ hash functions from $\mathcal{H}$, and record the index of a hash function that makes the data structure space-efficient. This way, the space bound can be made to be high-probability worst-case.

\begin{corollary}
    With high probability in $n$ over the choice of hash functions, the construction of the data structure succeeds, and the data structure uses at most $\log \binom{U}{n} + \poly \log n$ bits of space.
\end{corollary}

\paragraph{Concatenating spillover representations.}
At this point, the only remaining task is to concatenate the representations for each bucket. 
In \cref{subsection:concat}, we have proved \cref{lem:storing_representations}, which concatenates $L$ spillover representations $(m_i, k_i)$, where the length of $m_i$ and the spill universe $K_i$ of $k_i$ are determined by the size parameter $s_i$. For our use, we set $L = n / B$ to be the number of buckets, $(m_i, k_i)$ to be the spillover representation of bucket $i$, and $s_i$ to be the number of elements in bucket $i$.

Recall that in \cref{subsection:concat}, we have assumed that $m_i$ is a sequence of complete $w$-bit words, i.e., its length is a multiple of $w$ bits, because otherwise, we may merge the leftover bits into the spill $k_i$. Here, we further assume that the spill universe $K_i > L^3$, since otherwise, we may merge the last $O(1)$ words of $m_i$ into the spill again.

\cref{lem:storing_representations} also needs three extra parameters: $S$ denotes the number of possible values for $s_i$; $M_{\min}$ and $M_{\max}$ denote the lower and upper bounds of the length of $m_i$ (in words).
Since $s_i \in [B - B^{2/3},\, B + B^{2/3}]$, we set $S=2B^{2/3}+1$. By \cref{lem:bucket_rep}, the length of $m_i$ is $\log \binom V{s_i} \pm O(\log n)$ bits. Therefore, we have the following bounds:
\[
\begin{aligned}
    M_{\min} &= \frac1w\min_{s\in [B-B^{2/3},\,B+B^{2/3}]} \log \binom{V}{s} - O(1), \\
    M_{\max} &= \frac1w\max_{s\in [B-B^{2/3},\,B+B^{2/3}]} \log \binom{V}{s} + O(1).
\end{aligned}
\numberthis\label{definition_M_minmax}
\]

\cref{lem:storing_representations} additionally requires the parameters to satisfy \eqref{inequality_store_representations}, which we recall and prove below.

\begin{claim}
    \cref{inequality_store_representations} holds, i.e., $M_{\min} > \Theta\bk*{(M_{\max} - M_{\min}) \log (L M_{\max})+S\log L}.$
\end{claim}

\begin{proof}
Recall that we have assumed $U \ge 2n$, and thus $V \ge 2B$. We begin the proof by obtaining closed-form bounds for $M_{\min}$ and $M_{\max}$. We have
\[
M_{\min}=\frac1w \log\binom{V}{B-B^{2/3}}-O(1)
\]
by the monotonicity of binomial coefficients.
For $M_{\max}$, we first focus on a single $s\in [B-B^{2/3},\, B+B^{2/3}]$, and obtain
\[
\binom{V}{s}
= \binom{V}{B - B^{2/3}} \cdot \prod_{j = B - B^{2/3} + 1}^{s} \frac{V - j + 1}{j}
\leq \binom V{B-B^{2/3}}\cdot\bk*{\frac{V}{B-B^{2/3}}}^{2B^{2/3}}.
\]
Taking the logarithm, we get
\[
M_{\max} \le M_{\min} + \frac{1}{w}\log\frac{V}{B-B^{2/3}} \cdot 2B^{2/3}+O(1).
\]

We are now ready to prove \eqref{inequality_store_representations}. The left-hand side of \eqref{inequality_store_representations} can be bounded as
\begin{align*}
\textup{LHS of \eqref{inequality_store_representations}} = M_{\min} &= \frac 1w\log \binom{V}{B-B^{2/3}}-O(1)\\
&\geq\frac 1w\log\binom{2(B-B^{2/3})}{B-B^{2/3}}-O(1)\\
&\ge \frac 1w\bk*{2(B-B^{2/3})-O(\log B)}-O(1)\\
&=\Omega(B/\log n)-O(1),
\end{align*}
where the second line is because $V \ge 2B$; the third line is because $\binom{2n}{n} \ge 2^{2n} / n$.
For the right-hand side, clearly $L, M_{\min}, M_{\max} \le O(n)$, so
\begin{align*}
\textup{RHS of \eqref{inequality_store_representations}} ={} & \Theta\bk[\Big]{(M_{\max}-M_{\min})\log (L M_{\max})+S\log L}\\
\le {} & O\bk*{\frac1w B^{2/3}\log\frac{V}{B-B^{2/3}}\log (L M_{\max})+B^{2/3}\log L}\\
\le{} & O\bk*{\frac1w B^{2/3}\log^2 n+B^{2/3}\log n}\\
={} & O\bk*{B^{2/3}\log n} \\
\ll{} & \textup{LHS of \eqref{inequality_store_representations}},
\end{align*}
where the last inequality holds because we can choose $B$ to be a sufficiently large polylogarithmic factor of $n$. This concludes the proof.
\end{proof}

\paragraph{Removing the assumption of fully random hash functions.}
Up until this point, we have assumed that the data structure (more specifically, the augmented retrievals) has access to a fully random hash function. Now we show how to remove this assumption in the cell-probe model. One natural idea to do so is to compute a hash function that makes the construction algorithm successful, and then store it in the memory, but this method will consume too much space. We use the following approach instead. %

Recall that the total number of possible inputs is $\binom{U}{n}$. For each input, \cref{lem:storing_representations} states that there is a $1 - 1/\poly n$ probability (over the choice of the hash function) that the construction algorithm succeeds.
Therefore, if we sample a set $H$ of $\log \binom{U}{n} = \poly n$ hash functions, then with at least constant probability, for every possible input, there exists a hash function $h \in H$ such that the construction succeeds over this hash function. Hence, there exists a set $H$ satisfying this condition, and we fix $H$ to be the set satisfying this condition with the lexicographically smallest binary representation.

Since the value of $H$ is deterministic and does not depend on the key set we need to store, we can access $H$ for free in the cell-probe model. We store the index of a hash function in $H$ that makes the construction algorithm succeed on the current input, which takes $O(\log n)$ bits. The query algorithm already knows the set $H$, so after reading the index of the hash function, it can evaluate the hash function $h$ on any input by itself.
\end{proof}

\label{subsection:altogether}

\section{Generalizing to Word RAM}
\label{sec:word_ram}
From now on, we switch our focus to the word RAM model. The main result of this section is the following theorem, which is the version of \cref{thm:intro} on membership data structures.

\begin{theorem}
\label{thm:word_RAM_main}
In the word RAM model, there is a static membership data structure storing $n$ keys from a universe of size $U \in [2n,\, \poly n]$, using $\log \binom{U}{n} + \poly \log n$ bits of space, assuming access to a fixed lookup table of $n^\eps$ bits and a hash function that can be encoded using $n^\eps$ bits, such that:
\begin{itemize}
\item The construction algorithm succeeds with high probability in $n$.
\item The query algorithm runs in worst-case constant time.
\end{itemize}
\end{theorem}

Directly applying the algorithm in the cell-probe model gives us a dictionary in the word RAM model that has $O(\poly \log n)$ redundancy and answers queries in constant time, assuming access to a random hash function and a lookup table of size $O(n^{\epsilon})$. In this section, we will focus on removing the random hash function for the word RAM model.

First, we recall how randomness was used in our cell-probe data structure. It was used twice: once when we hash keys into buckets, and once when we construct the augmented retrieval. For the former one, we were already using a $\poly \log n$-wise independent hash function, so we can simply store the encoding of such a hash function (which takes $O(n^\eps)$ bits by \cref{lem:random_permutation}). The main difficulty lies in constructing the augmented retrieval. In this section, we construct an augmented retrieval that only accesses a hash function that can be encoded using $O(n^\eps)$ bits. When substituting this for the original augmented retrieval, we obtain the algorithm described in \cref{thm:word_RAM_main}.

Similar to \cref{section:retrieval}, here we also assume that values and the memory words are elements in a finite field $\F$. The following lemma is similar to \cref{lem:augmented_retrieval} with several modifications: The number $m$ of augmented elements has changed, and we assume that we only have limited randomness instead of a fully independent hash function. Also, we need a lookup table of size $\tilde{O}(n^\epsilon)$ to solve this augmented retrieval in the word RAM model.%
\footnote{It suffices to upper bound the size of the lookup table and the hash functions by $\tilde O(n^\eps)$, because $\eps$ is a small constant that we can choose arbitrarily. By choosing a smaller parameter $\eps' < \eps$, the space bound becomes $\tilde O(n^{\eps'}) \ll n^{\eps}$ as desired.}

\begin{lemma}
  \label{lem:better_augmented_retrieval}
  Let $\F$ be a finite field of size $\max(n^6,\, N+n^2)\le |\F|=\poly n$. In the word RAM model where the memory consists of words in $\F$, there is an augmented retrieval data structure that stores $n$ key-value pairs, where the keys are from a universe $[N]$ and values are in $\F$, as well as an augmented array $a_1,\dots,a_m$ of elements in $\F$ where $m= \Omega(N\log n+n\log^2n)$, such that:
  \begin{itemize}
  \item The data structure uses $n+m+\poly\log n$ memory words in $\F$ and can answer retrieval queries and augmented queries in constant time.
  \item The data structure assumes access to a hash function that can be encoded using $n^\eps$ bits, and succeeds with probability $\ge 1 - 1/n$ over this hash function.
  \item The data structure also assumes access to a fixed lookup table of size $\tilde{O}(n^\epsilon)$.
  \end{itemize}
\end{lemma}

\newcommand{\A}{\mathcal{A}}
\newcommand{\B}{\mathcal{B}}
\newcommand{\M}{\mathcal{M}}

\paragraph{Set-up and terminologies.}
Before we prove this lemma, let us first recall the high-level framework for the augmental retrieval in \cref{section:retrieval}, and set up terminologies for our discussion. There are $N$ potential retrieval queries and $m$ augmented queries; the data structure should correctly answer $n + m$ of them ($n$ specific retrieval queries plus all $m$ augmented queries), which we call the \defn{valid queries}. The augmented retrieval is solved by sampling a sparse $(N + m) \times (n + m)$ matrix $\M$, where each row represents a query, and the columns represent $n + m$ memory cells. For conciseness, we use \defn{retrieval rows}, \defn{augmented rows}, and \defn{valid rows} to refer to the rows in $\M$ that represent retrieval queries, augmented queries, and valid queries, respectively. Each query is answered by computing the inner product of the memory contents and the row in $\M$ representing that query.

The correctness of this approach relies on the fact that the submatrix consisting of all valid rows has full rank, which was proven by the Schwartz-Zippel lemma: A non-zero multivariate polynomial of degree $O(n)$ over a finite field of order $\poly n$ is non-zero with good probability when evaluated on a random input.
Unfortunately, this lemma does not work when the input only has limited independence (for instance, if the values are sampled from a $\poly \log n$-wise independent hash function). So we need a new approach to generate the sparse matrix.

On the technical level, the construction in \cref{section:retrieval} first generates a smaller, less sparse matrix $\A$ of size $N \times n$, where the $N$ rows only represent the retrieval queries. Each row of $\A$ has $O(\log n)$ non-zero entries at independently random locations. It is shown that the submatrix of $\A$ on the valid rows has full rank with constant probability. Next, a sparsification technique is used to generate the larger matrix $\M$, which takes $\A$ as the input. It ``moves'' the non-zero entries from $\A$ (retrieval rows) to the augmented rows, so that every row in $\M$ only has a constant number of non-zero entries.

Our data structure for \cref{lem:better_augmented_retrieval} almost follows the same framework with the following distinction. We also use the augmented queries in a new way in addition to sparsifying the matrix $\A$: We choose a subset of augmented queries and use them in the construction of matrix $\A$ ``as if'' they were retrieval queries. We call these queries \defn{filler queries} and let $n_f = O(n)$ be the number of them. As a result, the matrix $\A$ in our new construction will have size $(N + n_f) \times (n + n_f)$, where the rows represent the retrieval queries and the filler queries. The filler queries will benefit the construction of $\A$ because they are always \emph{valid}. Similar to before, after we construct matrix $\A$, we perform the sparsification step using the rest of the augmented queries, obtaining matrix $\M$.

In the remaining part of this section, we will focus on constructing such an $(N + n_f) \times (n + n_f)$ matrix $\A$ using only $\poly \log n$ random bits, such that the total number of non-zero entries in $\A$ is $O(N \log n + n \log^2 n)$, and for each subset $S \subset [N]$ of size $n$, the submatrix $\A_{S \cup (N, N + n_f], *}$ (i.e., the submatrix on valid queries) has full rank with probability $\ge 1 - 1/n$. As we discussed above, this will imply \cref{lem:better_augmented_retrieval}.

\paragraph{Tree of blocks.}

Our construction of matrix $\A$ is based on a tree of blocks. Each block on the tree will be assigned a set of rows and a set of columns, so that it represents a submatrix (a block) in the matrix $\A$. Additionally, each row in the matrix $\A$ (i.e., each retrieval query or filler query) will be assigned to a unique block.

Let $B \defeq O(\log n)$. The first step of our construction is to hash the \emph{retrieval queries} into $n/B$ \emph{level-1 blocks} using a $\poly \log n$-wise independent hash function, such that each level-1 block contains $B$ valid retrieval queries in expectation. We also define $\Delta_1 = c \cdot \sqrt{B \log n} = O(\log n)$ for a large constant $c$, so that with high probability in $n$, the number of valid retrieval queries in each level-1 block is within $B \pm \Delta_1$ by a Chernoff bound for $\poly \log n$-wise independent random variables \cite{schmidt1995chernoffhoeffding}. Each level-1 block is assigned $B + \Delta_1$ columns so that the number of columns in each level-1 block is not smaller than the number of (valid) rows.

Next, we group every two level-1 blocks together and create a common parent for them. The parent is called a \emph{level-2 block}. The number of valid retrieval queries hashed to a level-2 subtree (i.e., a subtree rooted at a level-2 block) is $2B \pm \Delta_2$ with high probability, where $\Delta_2 = \sqrt{2} \cdot \Delta_1$. Note that all retrieval rows are already assigned to level-1 blocks, so starting from level 2, we will only assign filler rows (queries) to the blocks. Each level-2 block is assigned $2\Delta_1 + \Delta_2$ filler rows. The columns assigned to each level-2 block $u$ include the columns of $u$'s two children, as well as $2\Delta_2$ additionally assigned columns which we call the \defn{supplementary columns} for block $u$. Again, when we focus on a level-2 subtree, the number of columns assigned to the subtree ($2B + 2\Delta_1 + 2\Delta_2$) is not smaller than the number of (valid) rows ($2B \pm \Delta_2 + 2\Delta_1 + \Delta_2$).

We repeat this process: On each level $i$, we group every two level-$(i - 1)$ blocks together and create a level-$i$ block as their common parent. Each level-$i$ block is assigned $2\Delta_{i-1} + \Delta_i$ filler rows and $2\Delta_i$ supplementary columns in addition to the columns already assigned to its children, where $\Delta_i \defeq \Delta_1 \cdot 2^{(i-1)/2} = c \cdot \sqrt{2^{i-1} \cdot B \log n} = O(2^{i/2} \log n)$ is a parameter ensuring that the number of valid retrieval queries hashed to a level-$i$ subtree is within $2^{i-1} B \pm \Delta_i$ with high probability. The same process is repeated until level $h = \log(n/B) + 1$, where there is only one level-$h$ block remaining, and we call it the \defn{root}. For the root level, we specially define $\Delta_h \defeq 0$, because the number of valid rows in the entire tree is fixed. The root block is then assigned $2 \Delta_{h-1}$ filler rows and no supplementary columns. It is easy to verify that the number of filler queries in the entire tree is $n_f = O(n)$, because the number of filler queries in each level decreases geometrically.

For the consistency of notations, we specially define the supplementary columns for a level-1 block to be all columns in the block.
Hence, every column belongs to a unique block on the tree as a supplementary column.
For the sake of discussion, we arrange all columns in matrix $\A$ in the postorder of the block containing them as a supplementary column. As a result, the columns assigned to each block is a consecutive interval of all columns $[1,\, n + n_f]$, and the supplementary columns occupy a suffix of the interval.
See \cref{fig:matrix_partition} for an example.

\begin{figure}[ht]
    \centering
    \begin{tikzpicture}
        \tikzmath{
            \layers = 4;
            \layergap = 0.8;
            \len = 8;
            \exlen = 1;
            function drawmatrix(\level, \id, \lx, \rx, \lastly, \lastry, \layerstarty, \layerheight, \layerextrawidth) {
                \curly = \layerstarty + \id * \layerheight;
                \curry = \curly + \layerheight;
                \lastlevel = int(\level - 1);
                \ymid = \layerstarty + \layerheight * pow(2, \layers - \level - 1);
                \yend = \layerstarty + \layerheight * pow(2, \layers - \level);
                if \id < 1 then {
                    {\node[] at (\len + 1, \ymid) (11) {Level $\level$};};
                    if \level >= \layers then {
                        {\path (\lx, \curly) --  (\lx, \curry) node[midway,left] {$2\Delta_\lastlevel$};};   
                    };
                    if \level <= 1 then {
                        {\path (\lx, \curly) --  (\lx, \curry) node[midway,left] {$B \pm \Delta_\level$};};
                        {\path (\lx, \curly) --  (\rx, \curly) node[midway,above] {$B + \Delta_\level$};};
                        {\draw [decorate,decoration={brace,amplitude=5pt,mirror,raise=4ex}] 
                        (-0.75, \layerstarty) -- (-0.75, \yend) node[midway,xshift=-3em]{$n$};};
                    };
                    if (\level > 1) * (\level < \layers) then {
                        {\path (\lx, \curly) --  (\lx, \curry) node[midway,left] {$2\Delta_\lastlevel + \Delta_\level$};};  
                        {\path (\rx - \layerextrawidth, \curly) --  (\rx, \curly) node[midway,above] {$2\Delta_\level$};};
                    };
                };
                if \level < \layers then {
                    {\draw [loosely dashed, -] (\lx, \curly) -- (\lx, \lastry);};
                    if mod(\id, 2) == 1 then {
                        {\draw [loosely dashed, -] (\rx, \curly) -- (\rx, \lastry);};
                    };
                };
                if \level <= 1 then {
                    {\fill [gray!50!white] (\lx, \curly) rectangle (\rx, \curry);};
                    {\fill [pattern=north east lines] (\lx, \curly) rectangle (\rx, \curry);};
                    {\fill [gray!50!white] (\lx, \lastly) rectangle (\rx, \lastry);};
                };
                if (\level > 1) * (\level < \layers) then {
                    {\fill [gray!50!white] (\rx - \layerextrawidth, \curly) rectangle (\rx, \curry);};
                    {\fill [pattern=north east lines] (\rx - \layerextrawidth, \curly) rectangle (\rx, \curry);};
                    {\fill [gray!50!white] (\rx - \layerextrawidth, \lastly) rectangle (\rx, \lastry);};
                };
                {\draw (\lx, \curly) rectangle (\rx, \curry);};
                \nextrx = ifthenelse(\level == \layers, \rx, \rx - \layerextrawidth);
                \nextmidx = (\lx + \nextrx) / 2;
                \nextstarty = \layerstarty + \layerheight * pow(2, \layers - \level) - \layergap * pow(0.5, \layers - \level);
                \nextheight = \layerheight / sqrt(2);
                \nextextra = \layerextrawidth / sqrt(2);
                if \level > 1 then {
                    drawmatrix(int(\level - 1), \id * 2, \lx, \nextmidx, \curly, \curry, \nextstarty, \nextheight, \nextextra);
                    drawmatrix(int(\level - 1), \id * 2 + 1, \nextmidx, \nextrx, \curly, \curry, \nextstarty, \nextheight, \nextextra);
                };
            };
            drawmatrix(\layers, 0, 0, \len, 0, 0, 0, -1, \exlen);
        }
    \end{tikzpicture}
    \caption{Tree of blocks with $h = 4$ levels. Every rectangle represents a block on the tree, in which the hatched area represents the supplementary columns of a block, and the gray area represents the locations of non-zero entries in the matrix.}
  \label{fig:matrix_partition}
\end{figure}
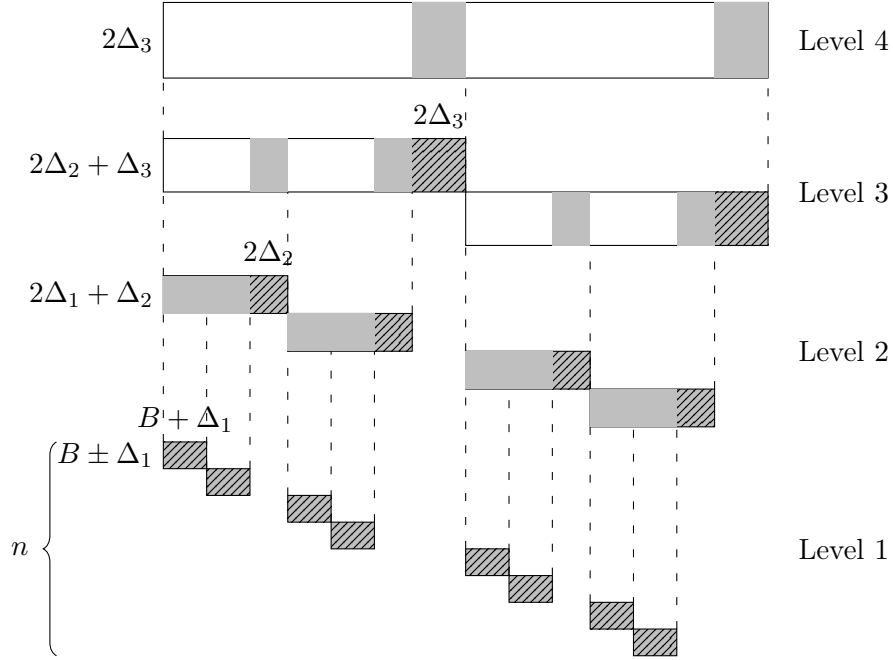

For each level-$i$ block $u$, let $\delta_u$ denote the difference between the number of columns and the number of valid rows in $u$'s subtree. We point out that $\delta_u \in [0,\, 2\Delta_i]$ with high probability: The lower bound $\delta_u \ge 0$ is already mentioned when we define the parameters for the tree; the upper bound follows as the number of columns in $u$'s subtree is fixed, while the number of valid rows only varies in a range of $2\Delta_i$ (i.e., $B \cdot 2^{i-1} \pm \Delta_i$) with high probability. In the remainder of this section, we will assume $\delta_u \in [0, 2\Delta_i]$ always holds, and define $F_u$ as the set of the last $\delta_u$ columns in block $u$.

\paragraph{Setting non-zero entries.}

Next, we specify the positions of non-zero entries based on the tree of blocks. We define the set of \defn{non-zero columns} of a level-$i$ block $u$ as the union of supplementary columns of $u$ itself and both of $u$'s \emph{direct} children (not all descendants). In particular, the non-zero columns of a level-1 block are simply all columns within that block. An entry in matrix $\A$ is non-zero if and only if it is on a non-zero column of the block containing that entry.

The matrix $\A$ is then constructed as follows:
\begin{itemize}
\item Randomly sample $h$ non-zero elements in $\F$, denoted by $x_1,\dots,x_h$.
\item For each query (row) $v\in [N+n_f]$ (either retrieval or filler), suppose the unique block $u$ containing it is on level $i$ (recall that even the invalid retrieval queries are assigned to some block, so this is well defined). For every non-zero column $j$ of block $u$, we set $A_{v,j} = (x_i)^j / (v + j)$.
\end{itemize}
With this construction, we can verify that the number of non-zero entries is small.

\begin{claim}
    \label{clm:nonzero_entry}
    The number of non-zero entries in $\A$ is $O(N\log n+n\log ^2n)$.
\end{claim}

\begin{proof}
    Each retrieval query has $B + \Delta_1 = O(\log n)$ non-zero entries on its row, so they contribute $O(N\log n)$ non-zero entries in total.
    
    On level $i$ ($1<i<h$), there are 
    \[
    (2\Delta_{i-1} + \Delta_i) \cdot \frac{n}{2^{i-1} B} = 
    O(2^{i/2} \log n) \cdot \frac{n}{2^{i-1} \log n} =
    O(n/2^{i/2})
    \]
    filler queries. In particular, on level-$h$, there are $2\cdot \Delta_{h-1}=O(n/2^{h/2})$ filler queries as well.
    
    For each level-$i$ block ($i > 1$), the number of non-zero columns is $2\Delta_{i-1} + 2\Delta_i = O(\Delta_{i-1}) = O(2^{i/2} \log n)$. So for every level $i > 1$, the total number of non-zero entries contributed by filler queries of level $i$ is $O(n \log n)$. Multiplied by the number $\log n$ of levels, we know the number of non-zero entries in $\A$ is $O(N \log n + n \log^2 n)$.
\end{proof}

Now consider the values of the entries. Our construction has the property that the submatrix formed by the rows and non-zero columns of a block is a Cauchy matrix with an extra $(x_i)^j$ term multiplied on each column. A useful property of Cauchy matrices is that any square submatrix has full rank:

\begin{lemma}[{\cite[Theorem 5.1]{Blmer1995AnXE}}]
    \label{lem:cauchy}
    Let $M \in \F^{m\times m}$ be a matrix with $M_{i,j} = 1/(a_i-b_j)$, where the values $a_i, b_j$ are pairwise distinct. Then $M$ has full rank.
\end{lemma}

The Cauchy matrix we constructed has parameters $a_i \in [1,\, N+n_f],\, b_j \in [-n-n_f,\,-1]$. Since we require that $|\F|>N+n^2$, those values are distinct in $\F$, so the lemma applies.

\paragraph{Analysis of correctness.} We will show that $\A_{S \cup (N, N + n_f], *}$ (i.e., the submatrix of $\A$ consisting of all valid rows) has full rank.

\begin{lemma}
    \label{lem:nonsingularity}
    For each set $S \subset [N]$ of size $n$, the submatrix $\A_{S\cup (N, N + n_f],*}$ has full rank with probability $\ge 1 - 1/n$.
\end{lemma}

Before we prove this lemma, we first assume it holds and continue to prove the main conclusion (\cref{lem:better_augmented_retrieval}) of this section.

\begin{proof}[Proof of \cref{lem:better_augmented_retrieval}]
    As we discussed earlier, given the matrix $\A$, the main task remaining is to sparsify $\A$. We apply the same approach as in \cref{section:retrieval}: For a row with $t$ non-zero entries, we assign $t$ augmented queries (rows) to sparsify it. Since there are $O(N\log n+n\log^2n)$ non-zero entries in $\A$, we have a sufficient number of augmented queries to sparsify the matrix.

    \smallskip

    To allow fast queries to augmented elements, given an augmented query, we need to decide the row and the non-zero entry in the row that it is assigned to in constant time. This can be done with the following careful assignment.
    
    Recall that the number of non-zero entries of a row is the number of non-zero columns of its block, which only depends on the level of the block.
    Each row in level 1 has $B + \Delta_1$ non-zero elements, so
    we assign the first $N (B + \Delta_1)$ augmented queries to level 1, each sparsifying one non-zero entry.
    For every other level, we assign a fixed number of $q \defeq \Theta(n \log n)$ augmented queries, in which a prefix of queries are used to sparsify the rows. (We have already calculated in the proof of \cref{clm:nonzero_entry} that each level contains $O(n \log n)$ non-zero entries, so $q$ augmented queries suffice.)
    Then, given the index of an augmented query, we can compute in constant time the level it corresponds to, as well as the exact position of the non-zero entry assigned to this augmented query.

    \smallskip

    Next, we focus on computing the value of the non-zero entries. Field addition and multiplication can be done in constant time given a lookup table of size $O(n^\eps)$. Apart from those operations, we also need to compute the power and the multiplicative inverse of some value, since every non-zero entry in $\A$ is of the form $x_i^j/(v+j)$. 
    
    In order to compute the power $x_i^j$, we first fix a primitive root $g$ of $\F$. In the sampling process, instead of directly sampling a non-zero $x_i$, we sample a number $a_i\in [1,|\F|)$ and let $x_i = g^{a_i}$. Then the task becomes to compute $x_i^j = g^{a_i\cdot j}$. We store the values $g^{i \cdot (n^{\eps})^j}$ in a lookup table for every $0 \le i < n^{\eps}$, $0 \le j < O(\log |\F| / (\eps \log n)) = O(1)$. Given this lookup table, when we need to compute $g^{a_i \cdot j}$, we write down $a_i \cdot j$ in the base-$n^{\eps}$ representation: $a_i \cdot j = \sum_{k} v_k \cdot (n^{\eps})^k$, then compute the product of the corresponding table entries $g^{v_k \cdot (n^{\eps})^k}$. The lookup table consists of $O(n^{\eps})$ words, and computing each power of $g$ takes $O(1/\eps) = O(1)$ time given the lookup table.

    Computing the multiplicative inverse of $v+j$ is more complicated, but we can circumvent this issue by multiplying the other non-zero entries in the row by $v+j$ instead. Specifically, if the non-zero entries of an augmented query (after sparsification) are $x_i^j/(v+j)$, $1$, and $-1$, then we will multiply this row by $v+j$, i.e., the actual coefficients used to answer this query will be $x_i^j, v+j$, and $-(v+j)$. This way, we don't have to compute multiplicative inverses to answer this query. This modification does not change the rank of the matrix.

    \smallskip

    Finally, we analyze the redundancy of the augmented retrieval. The redundancy comes from three sources: the $\poly\log n$-wise independent hash function, the random variables $x_1,\dots,x_h$, and the lookup table used to compute $x_i^j$. The last two can be stored in $\tilde O(n^\eps)$ bits of space, as we explained in the previous paragraphs.
    
    For the $\poly \log n$-wise independent hash function, we use the construction from \cite{thorup2013simple}.
    
    \begin{lemma}[{\cite[Corollary 3]{thorup2013simple}}]\label{lem:hash_function}
    There is a hash function from universe $[U]$ to range $[R]$ that uses $o(U^{\eps})$ bits of space, can be evaluated in $O(1/\eps)$ time, and is $U^{\Omega(\eps^2)}$-wise independent with universal probability $1 - o(1 / U^\eps)$.
    \end{lemma}
    
    Using the lemma, we can construct a hash function with more than $\poly\log n$-wise independence that uses $O(n^{\epsilon})$ bits of space, with success probability $1-1/n^\eps$. In order to boost the success probability from $1-1/n^\eps$ to $1-1/\poly n$, we can sample $\Theta(1/\eps)$ hash functions using \cref{lem:hash_function}, then use their sum as the output of our hash function. As long as one of the functions is $\poly\log n$-wise independent, so is the sum of all functions.
\end{proof}

Now, we continue to prove \cref{lem:nonsingularity}.

\begin{proofof}{\cref{lem:nonsingularity}}
    We first recall that, for each block $u$, $\delta_u$ is the difference between the number of columns and valid rows in $u$'s subtree; $F_u$ is defined as the last $\delta_u$ columns in $u$, which is (with high probability) a subset of $u$'s supplementary columns.

    We prove the lemma by induction with the following induction hypothesis at level $i$ ($1 \le i \le h$).
    \begin{hypothesis}
        Let $M$ be the submatrix spanned by all the valid rows (queries) and the columns assigned to $u$'s subtree.
        With probability $\ge 1 - i/n^2$ over the choices of $x_1, \dots, x_i$, for each level-$i$ block $u$, if we remove the columns $F_u$ from $M$, the resulting square matrix has full rank.
    \end{hypothesis}

    When $i = h$, the hypothesis implies \cref{lem:nonsingularity}.

    The base case $i=1$ is simple: For each level-1 block $u$, the submatrix we obtain is a Cauchy matrix, with each column multiplied by a power of $x_1$. Since $x_1 \ne 0$, this matrix is row independent.

    Now, we fix some $1 < i \le h$, and assume the hypothesis holds for $i - 1$. We fix a level-$i$ block $u$ and denote its children by $t_1, t_2$. Let the matrix in hypothesis be $M$. For $k \in \BK{1, 2}$, we let $Q_k, C_k$ be the set of valid rows (queries) and columns in the subtree of $t_k$, respectively; we let $Q_0$ be the rows in block $u$ and let $C_0$ be the supplementary columns of $u$. By definition, $\BK{Q_0, Q_1, Q_2}$ form a partition of the rows of $M$, while $\BK{C_1, C_2, C_0}$ form a partition of the columns.
    We will show that, assuming the induction hypothesis for level $i-1$, the hypothesis holds for $u$ with probability $\ge 1 - 1/n^3$. Then, taking a union bound over all the level-$i$ blocks would imply the hypothesis for level $i$.

    \smallskip

    We prove the nonsingularity by Gaussian elimination. First, we randomly sample $x_1,\dots,x_{i-1}$ and view each entry in the matrix as a polynomial in $x_i$. Then the entries on rows in $R_1$ and $R_2$ are constants. By the induction hypothesis, we can perform elementary row operations on each submatrix $M_{Q_k, C_k}$ to transform its first $|Q_k|$ columns to an identity matrix (see \cref{fig:elimination} (b)).

    \begin{figure}[ht]
      \newcommand{\subwidth}{0.22\textwidth}
      \newcommand{\subgap}{\hspace{-0.7em}}
      \newcommand{\subarrow}{\Rightarrow}
      \[
        \begin{matrix}
          \matwrap{\includegraphics[width=\subwidth]{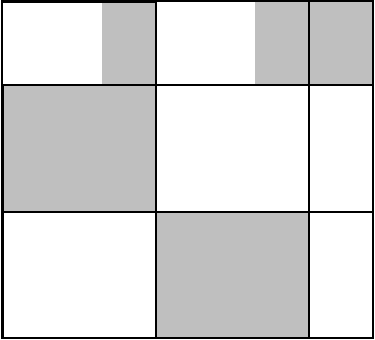}}
          & \subgap\subarrow
          & \subgap\matwrap{\includegraphics[width=\subwidth]{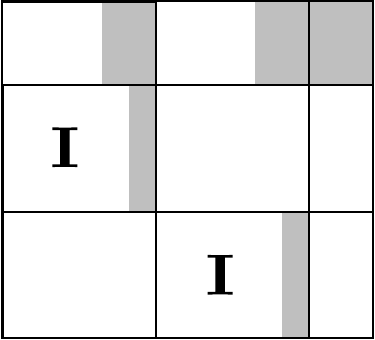}}
          & \subgap\subarrow
          & \subgap\matwrap{\includegraphics[width=\subwidth]{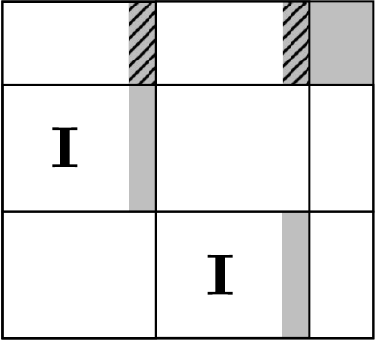}}
          & \subgap\subarrow
          & \subgap\matwrap{\includegraphics[width=\subwidth]{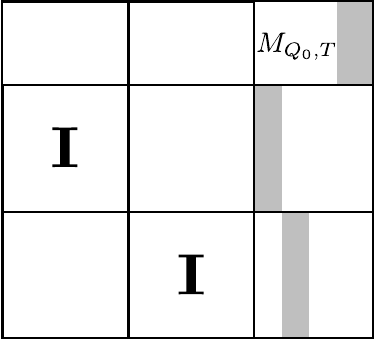}}
          \\
          \textup{(a)} &&\subgap \textup{(b)} &&\subgap \textup{(c)} &&\subgap \textup{(d)}
        \end{matrix}
      \]
      \caption{Gaussian elimination on matrix $M$.
        (a) The matrix is partitioned into $3 \times 3$ blocks, where they correspond to row sets $Q_0, Q_1, Q_2$ from top to bottom, and column sets $C_1, C_2, C_0$ from left to right. The third column set $C_0$ consists of the $2\Delta_i$ supplementary columns of block $u$. Gray area of the matrix represents unknown non-zero entries.
        (b) Perform elimination within each child's submatrix, obtaining two identity submatrices.
        (c) Use two children's submatrices to eliminate entries on $Q_0$ above the identity parts. The non-zero entries in the hatched area changed during this step.
        (d) Permute the columns. We show that $M_{Q_0, T}$ has full rank.
      }
      \label{fig:elimination}
    \end{figure}

    Then, we use the submatrices $M_{Q_k, C_k}$ ($k \in \BK{1, 2}$) to eliminate some entries on rows in $Q_0$ via elementary row operations, ensuring that rows in $Q_0$ have zero values on the columns occupied by the identity matrices (see \cref{fig:elimination} (c)).

    At this point, each row in $Q_0$ is only non-zero on the columns $F_{t_0} \cup F_{t_1}$, and $C_0$---the $2\Delta_u$ supplementary columns of $u$. In addition, for each non-zero entry on rows $v \in Q_0$, suppose it had value $x_i^j/(v+j)$ before the row operations. Then currently its value is still a polynomial of degree $j$, and the coefficient of $x_i^j$ is still $1/(v+j)$. This is because the previous elimination steps only added a polynomial of degree $<j$ to this entry, as can be easily verified.

    Let $T$ be the union of the columns $F_{t_0} \cup F_{t_1}$ and the first $2\Delta_i - |F_u|$ columns in $C_0$. In particular, when $i = h$, $T$ is just the union of $F_{t_0} \cup F_{t_1}$ since $\Delta_h = 0$. By permuting the columns, we see that it remains to show that $M_{Q_0, T}$ has full rank (see \cref{fig:elimination} (d)). Also, by definition of $|F_u|$, $M_{Q_0, T}$ is a square matrix.

    View $\det(M_{Q_0,T})$ as a polynomial in $x_i$. Its degree is $d \defeq\sum_{j\in T}j=O(n^2)$, and the coefficient of its highest term $x_i^d$ is the determinant of a Cauchy matrix, which is non-zero by \cref{lem:cauchy}. Therefore, $\det(M_{Q_0, T})$ is a non-zero polynomial; by the Schwartz-Zippel lemma, with probability $1-O(1/n^4)$ over the choice of $x_i$, $\det(M_{Q_0, T}) \ne 0$ and thus the induction hypothesis holds for block $u$. A union bound over all the level-$i$ blocks concludes the induction hypothesis, which further implies the lemma.
\end{proofof}

\paragraph{Extending to key-value dictionaries.}

So far, we have established a space-efficient \emph{membership} dictionary, which only answers if the queried key is in the key set or not. Next, we show that it is easy to extend this result to \emph{key-value} dictionaries, which also need to return a value associated with the queried key if the queried key is in the key set.

Our algorithm for membership dictionaries consists of two parts. After hashing all keys to buckets of size $B$, the first part (\cref{subsection:intra_bucket}) applies existing techniques from \cite{patrascu2008succincter,yu2020nearly} to encode the key set within each bucket using a spillover representation with only $O(1/n^2)$ bits of redundancy. Then, the second part (\cref{subsection:concat,sec:word_ram}) concatenates all spillover representations together without losing much space.
When we switch to key-value dictionaries, only the first part needs to be changed. To encode the key set within each bucket together with their associated values, we again apply black-box tools from \cite{yu2020nearly}, which achieve the same redundancy as the key-only case, and lead to a variant of \cref{lem:bucket_rep} that also stores associated values.
This, combined with the other parts of our algorithm, implies \cref{thm:intro} as a corollary. (Recall that $\OPT \defeq \log \binom{U}{n} + n \log \sigma$.)

\MainTheoremIntro*

\section{Discussion and Open Questions}

An important observation of this paper is that, although we do not know how to solve the retrieval problem with little redundancy, this problem becomes much simpler when we store the retrieval data structure together with an augmented array. We believe that this phenomenon will be useful for designing other data structures, and is worth further investigation.

Formally, we define the \defn{augmented redundancy} of a data structure to be the redundancy of storing this data structure together with an augmented array of arbitrary polynomial length.
Note that the augmented array is only \emph{informationally encoded} in the data structure, but is not required to be efficiently accessed.
This is a weaker requirement than what we achieve for the augmented retrieval problem in \cref{section:retrieval}.
Based on this notion, we propose the following open questions.

\paragraph{Separating augmented redundancy and (regular) redundancy.}

It is clear that the augmented redundancy of any data structure problem is not greater than its regular redundancy. However, it is an open question if they are always equal, i.e., whether there exist data structure problems for which, under the same query time constraint, the optimal augmented redundancy is (provably) strictly better than the regular redundancy.

For several problems such as dictionaries and retrieval data structures, there is a gap between the \emph{best-known} augmented redundancy and regular redundancy, but this is not sufficient to form a theoretical separation between the two concepts. Similar gaps also exist in lower bounds: For several important problems such as \emph{permutation}\footnote{The premutation problem requires us to store a permutation $\pi$ of the set $[n]$, supporting efficient queries of $\pi(i)$ and $\pi^{-1}(i)$ for each $i \in [n]$.} and \emph{access/select}\footnote{The access/select problem requires us to store a string of length $n$ with alphabet size $|\Sigma| = \poly n$, supporting efficient queries of \emph{access} (returning the $i$-th character in the string) and \emph{select} (reporting the $i$-th occurrence of a specific symbol $c$ in the string).} \cite{golynski2009cell}, the known lower bounds on regular redundancy do not naturally generalize to augmented redundancy.

\paragraph{Additivity of augmented redundancy.}

The second open question is the optimal redundancy of jointly storing \emph{multiple} independent data structures.
This problem is natural and fundamental, but surprisingly challenging, as directly concatenating the data structures may not give the optimal result---earlier discussion suggests that storing some data structure together with an independent (augmented) array may improve its redundancy, which is captured by the augmented redundancy.
However, we conjecture that using augmented arrays is the only way to reduce redundancy when we jointly store multiple data structures, i.e., the augmented redundancies for individual data structures add up to that for the joint data structure (asymptotically), stated as follows.

\begin{conjecture}
Assume $D_1, \ldots, D_k$ are $k$ (independent) data structure problems, where $D_i$ answers queries in $t_i$ time and has an optimal augmented redundancy $R_i$. For a joint data structure $A$ that can answer queries for $D_i$ in $O(t_i)$ time for all $i \in [k]$, we conjecture the augmented redundancy for $A$ is at least $\Omega\bk[\big]{\sum_{i=1}^k R_i}$.
\end{conjecture}

Suppose a data structure has multiple independent components, then if the conjecture is true, proving an augmented redundancy lower bound for each component is sufficient to derive an augmented redundancy lower bound (hence, a redundancy lower bound) for the whole data structure.

\paragraph{Augmented redundancy for classical problems.}

It is worth studying the optimal augmented redundancies for classical problems, especially whether they are equal to the regular redundancies. For some problems such as \emph{rank/select} \cite{patrascu2010cellprobe,viola2023new} and \emph{range minimum queries (RMQ)} \cite{liu2020lower,liu2022nearly}, the proofs of the nearly-tight lower bounds for regular redundancy naturally generalize to augmented redundancy, which implies that their augmented and regular redundancies are (almost) equal. For other classical problems including \emph{permutation} and \emph{access/select}, the known lower bounds \cite{golynski2009cell} do not generalize, and it remains as open questions to determine their augmented redundancies.
Moreover, it is also interesting to find other classical problems where the augmented array helps reduce the redundancy (like the retrieval problem), and see if these results help construct data structures with lower \emph{regular} redundancy (like in this paper).

\bigskip

The following two open questions are unrelated to augmented redundancy, and arose from our approach of static dictionary.

\paragraph{Redundancy of static retrievals.} 

In \cref{lem:better_augmented_retrieval}, we presented a static (augmented) retrieval data structure with $\poly\log n$ bits of redundancy which stores values of size $O(\log n)$ bits and supports constant-time queries, assuming we have an augmented array. It is of great interest to see whether we can achieve nice time and space bounds without using the augmented array.

Previous works on static retrieval data structures have made significant progress in the regime where the values are small: For instance, \cite{dietzfelbinger2019constanttime} presented a retrieval data structure which, when the values have $O(r)$ bits, has $O(\log n)$ redundancy and answers queries in $O(r)$ time. This algorithm is efficient when the values are of constant size, but is too slow for values of larger size. 
In the case where the values have $O(\log n)$ bits, even achieving $o(n)$ bits of redundancy and constant query time is open.

\paragraph{Construction time of static dictionaries.}
The construction time of our word RAM dictionary is $\poly n$, where the bottleneck is solving the linear equations in the retrieval data structure. In comparison, previous dictionaries \cite{patrascu2008succincter}, \cite{yu2020nearly} only require near-linear construction time. It is interesting to see whether the construction time can be improved.

\bibliographystyle{alpha}
\bibliography{reference.bib}

\appendix
\section{Proof of Lemma \ref{lem:random_matrix_rank}}
\label{app:random_matrix_rank}
In this appendix, we prove \cref{lem:random_matrix_rank}.

\matrixRank*

\begin{proof}
    We define a $n\times n$ matrix $P$ of polynomials in variables $X_{i,j}$ ($1\le i\le U$, $1\le j\le n$): $P_{i,j} = X_{i,j}$ if the entry $(i,j)$ is chosen in the sampling process, and $P_{i,j}=0$ otherwise. We can see that sampling $M$ is the same as first sampling $P$, then replacing each variable $X_{i,j}$ by a uniform random element in $\F$.

    We first show that $P$ is row independent with good probability, then the result for $M$ follows. This is because $\det(M)$ is the evaluation of $\det(P)$ on a random input, which is non-zero with probability $1-(\deg \det(P))/|\F|$ from the Schwartz-Zippel lemma (e.g., \cite[Theorem 7.2]{Motwani_Raghavan_1995}).

    To show the nonsingularity of $P$, we only have to show that the $n \times n$ bipartite graph $G = L \dot\cup R$ with $(\ind[P_{i,j} \ne 0])_{i,j}$ as its incidence matrix has a perfect matching, since a perfect matching contributes a monomial to the determinant that cannot be cancelled out. To show this, we recall Hall's theorem \cite{Hall1935OnRO}, which states that an $n \times n$ bipartite graph $L \dot\cup R$ has a perfect matching if and only if for every subset $S\subset L$, its neighborhood has size no smaller than itself, i.e., $|N(S)|\ge|S|$.

    We now fix two sets $S,T\subset [n]$ such that $|S|=|T|+1$, and study the probability that $N(S)\subset T$. This happens with probability $(|T|/n)^{|S|t}$. If we don't have a perfect matching, then $N(S)\subset T$ for at least one such pair of sets $S,T$, so a union bound over all the set pairs upper bounds the error probability:
    \begin{align*}
        \Pr[\text{no perfect matching}]&\le \sum_{i=1}^{n}\binom{n}{i}\binom{n}{i-1}\bk*{\frac{i-1}{n}}^{it}
    \end{align*}
    We split the summation into two parts: $i< n/2$ and $i\ge n/2$, and bound them separately. For the first part:
    \begin{align*}
        \sum_{i<n/2}\binom{n}{i}\binom{n}{i-1}\bk*{\frac{i-1}{n}}^{it} &\le \sum_{i<n/2}n^{2i-1}\bk*{\frac{1}{2}}^{it} \\
        &\le \sum_{i<n/2}n^{2i-1}\bk*{\frac{1}{n}}^{10i}\le \frac{1}{n}.
    \end{align*}
    For the second part (recall that $\bk*{\frac{n-1}{n}}^n<\frac{1}{e}$ for any positive $n$):
    \begin{align*}
        \sum_{i\ge n/2}\binom{n}{i}\binom{n}{i-1}\bk*{\frac{i-1}{n}}^{it} 
        &\le \sum_{i\ge n/2}n^{2n-2i+1}\bk*{\frac{i-1}{i}}^{it}\bk*{\frac{i}{n}}^{it} \\
        &\le \sum_{i\ge n/2}n^{2n-2i+1}n^{-10 \log e}\bk*{\frac{i}{n}}^{it} \\
        &\le n^{2n-9}\sum_{i\ge n/2}n^{-2i}\bk*{\frac{i}{n}}^{it}.
        \numberthis \label{eq:second_part_app}
    \end{align*}
    Let $f(i)=n^{-2i}(\frac{i}{n})^{it}$, and compare two adjacent terms $f(i),f(i+1)$ for $i\ge n/2$:
    \begin{align*}
        \frac{f(i)}{f(i+1)}&=n^2\bk*{\frac{i}{i+1}}^{(i+1)t}\bk*{\frac{n}{i}}^{t}\le n^2\bk*{\frac{2}{e}}^{t} < 1.
    \end{align*}
    Therefore, $f(i)$ is increasing in $i$. Since $f(n)=n^{-2n}$, we have
    \begin{align*}
        \textup{\cref{eq:second_part_app}}\le n^{2n-9}\cdot n\cdot n^{-2n}\le \frac{1}{n}.
    \end{align*}
    So the overall probability of $P$ being singular is $o(1)$.
\end{proof}

\section{Proof of Lemma \ref{lem:random_permutation}}
\label{app:random_permutation}
In this appendix, we prove \cref{lem:random_permutation}.

\randomPermutation*

Our construction is the same as Lemma 5.1 of \cite{li2024dynamic}, with the parameters being slightly different.

\begin{lemma}[Lemma 5.1 of \cite{li2024dynamic}, modified]
    Let $s,L$ be integers where $1\le s<L$. There is a family of hash functions $\mathcal{H}$ in which every member $h:\{0,1\}^L\to\{0,1\}^L$ is a bijection, satisfying:

  \begin{enumerate}[label=\textup{(\alph*)}]
  \item For any $h \in \mathcal{H}$ and any input $x \in \BK{0, 1}^L$, $h(x)$ and $h^{-1}(x)$ can be evaluated in $O(1)$ time.
  \item It takes $O(2^{\eps L})$ bits to store an $h \in \mathcal{H}$, where $\eps=\Omega\bk[\Big]{\sqrt{\frac{\log L}{L-s}}}$.
  \item\label{enum:no_overflow_in_lemma} For $n \ge 2^{s} \cdot s^4$ different inputs $x_1, \ldots, x_n$, if we divide $h(x_1), \ldots, h(x_n)$ into equivalent classes according to the first $s$ bits of $h(x_i)$, then with probability $\ge 1 - \frac{1}{4n^2}$, the number of elements in any equivalent class is between
    $\Bk{ \frac{n}{2^{s}} - \frac{1}{5} \bk{\frac{n}{2^{s}}}^{2/3}, \, \frac{n}{2^s} + \frac{1}{5} \bk{\frac{n}{2^s}}^{2/3}}$.
  \end{enumerate}
\end{lemma}

In their original lemma, it is required that $s\le (1-\Omega(1))L$, and $\eps$ is a constant independent of $s$ and $L$. However, their proof can be easily generalized to the case where $L - s$ is smaller for the following reason.
In their proof, the only place that uses the condition $s\le (1-\Omega(1))L$ is when showing that a $O(2^{(L-s)\eps^2/4})$-wise independent hash function is independent over any subset of $k\defeq (10\ln 2)L$ elements, i.e., that $k=O(2^{(L-s)\eps^2/4})$. When $L-s=\Omega(L)$, this easily holds. However, when we only have $s<L$, then in order for this to hold, we need to set $\eps=\Omega(\sqrt{\log L/(L-s)})$.

When we apply this to \cref{lem:random_permutation}, we need to set $L=\log U$ and $s=\log (n/B)$, then $L - s = \log (UB/n) = \Omega(\log B)$. Therefore, setting $\eps=\Omega(\sqrt{\log \log U/\log B})$ suffices.

\end{document}